\documentclass[aps,prx,floatfix,reprint,superscriptaddress]{revtex4-1}

\usepackage{graphicx}
\usepackage{amsmath}
\usepackage{amsfonts}
\usepackage{amssymb}
\usepackage{hyperref}
\usepackage{mathrsfs}
\usepackage{bbm}

\usepackage{palatino}

\newcommand{\uba}{Departamento de F\'\i sica, FCEyN, UBA, Pabell\'on 1,
  Ciudad Universitaria, 1428 Buenos Aires, Argentina}
\newcommand{\ifiba}{Instituto de F\'\i sica de Buenos Aires, UBA CONICET,
  Pabell\'on 1, Ciudad Universitaria, 1428 Buenos Aires, Argentina}

\usepackage{amsthm}
\newtheorem{Theorem}{Theorem}

\DeclareMathOperator{\tr}{tr}

\newcommand{\dyad}[1]{\left\lvert{#1}\middle\rangle\!\middle\langle{#1}\right\rvert}

\newcommand{\R}{\mathbb{R}}
\newcommand{\Ttransform}{\xrightarrow{\mathcal{E}}}
\newcommand{\critico}{\mathcal{R}^*\text{-state}}
\newcommand{\criticos}{\mathcal{R}^*\text{-states}}
\newcommand{\Wmin}{\mathcal{W}_{\rm form}^{\rm{}}}
\newcommand{\Wirr}{\mathcal{W}_{\rm irr}^{\rm{}}}
\newcommand{\Wext}{\mathcal{W}_{\rm ext}^{\rm{}}}
\newcommand{\Wform}{W_{\rm form}}
\newcommand{\WE}{W_{\rm ext}} 

\newcommand{\optimal}{\rho^{(N)}_{\text{min}}}
\newcommand\mean[1]{\mathinner{\langle{#1}\rangle}}
\newcommand\cwork{$c$-work of formation }
\newcommand\cworkns{$c$-work of formation}

\newcommand{\deit}[1]{{#1}}
\newcommand{\nsm}[1]{{#1}}
\newcommand{\kBT}{k_{\nsm{B}}T}
\newcommand{\tauS}{\tau_{\nsm{S}}}
\newcommand{\ZS}{Z_{\nsm{S}}}
\newcommand{\ZSN}{Z_{\nsm{S}}^{N}}
\newcommand{\HS}{H_{\nsm{S}}}
\newcommand{\system}{\nsm{S}}
\newcommand{\gS}{g_{\nsm{S}}}
\newcommand{\ES}{\mathcal{E}_{\nsm{S}}}
\newcommand{\eS}{E_{\nsm{S}}}

\newcommand{\DF}{\Delta F}
\newcommand{\norm}[1]{\left\lVert#1\right\rVert}
\usepackage{mathtools}
\DeclarePairedDelimiter\floor{\lfloor}{\rfloor}

\begin{document}

\title{Correlations as a resource in quantum thermodynamics}

\author{Facundo Sapienza} \affiliation{\uba}
\author{Federico Cerisola} \email{cerisola@df.uba.ar} \affiliation{\uba} \affiliation{\ifiba}
\author{Augusto J. Roncaglia} \email{augusto@df.uba.ar} \affiliation{\uba} \affiliation{\ifiba}



\begin{abstract}

The presence of correlations in physical systems can be a valuable resource for
many quantum information tasks. They are also relevant in thermodynamic transformations, 
and their creation is usually associated to some energetic cost. 
In this work, we study the role of correlations in the thermodynamic process of state formation in the single-shot regime, 
and find that correlations can also be viewed as a resource.
First, we show that the energetic cost of creating multiple copies of a given state can be reduced by allowing correlations in the final state. We obtain the minimum cost for every finite number of subsystems, and then we show that this feature is not restricted to the case of copies. More generally, we demonstrate that in the asymptotic limit, by allowing a logarithmic amount of correlations, we can recover standard results where the free energy quantifies this minimum cost.
\end{abstract}


\maketitle


\section{Introduction}


Quantum thermodynamics is a growing field aiming to extend 
thermodynamics to the limit of few number of systems in the quantum domain \cite{goold2016role,Vinjanampathy2016}.
This quest has been motivated by the theoretical interest in understanding the
fundamental limitations of thermodynamic transformations, and from a
practical point of view it has been driven by the current technologies
that allow to reach an incredible level of control of individual quantum systems.
Among the different approaches that have been put forward to analyze
 thermodynamics in this regime, a recent perspective to study nonequilibrium transformations of small number of systems
 in contact with a thermal bath, the so-called \deit{resource theory of
thermodynamics} \cite{janzing2000thermodynamic,brandao2013resource,Horodecki2013,Aberg2013}, 
has gained a lot of interest \cite{Lostaglio2015,egloff2015measure,halpern2015introducing,YungerHalpern2016,Richens2016,cwiklinski2015limitations,Gemmer2015,kwon2018clock,Lostaglio2015a,Lostaglio2017,lostaglio2018elementary,Masanes2017,mueller2017correlating,narasimhachar2015low,ng2015limits,Sparaciari2017a,VanderMeer2017,Sparaciari2017}. This framework captures the fundamental concepts of
thermodynamics with an operational approach to physics~\cite{Coecke2016}: by defining a set of operations an agent is
allowed to perform on a physical system, it characterizes the set of attainable transformations.
 The resource-theoretic approach to thermodynamics, while consistent with classical thermodynamics, has interesting properties
that depart significantly from the standard framework. 
In fact, in the single-shot regime thermodynamic transformations must satisfy a family of constraints
\cite{Brandao2015}, including the standard second law as a particular case. 
Furthermore, it naturally leads to a fundamental notion of \deit{irreversibility}, since in general
the amount of deterministic work required to perform a given transformation is greater than the work that can be drawn 
from the reverse process \cite{Horodecki2013,Aberg2013}. 

One of the main challenges in this field is to elucidate the role of properties such as quantum coherences and correlations in thermodynamic transformations.
Several works address the influence of coherence \cite{Lostaglio2015a,lostaglio2015quantum,cwiklinski2015limitations,korzekwa2016extraction,kwon2018clock} 
and correlations~\cite{jennings2010entanglement,del2011thermodynamic,bera2017generalized,friis2016energetics,perarnau2015extractable,Lostaglio2015,huber2015thermodynamic,bruschi2015thermodynamics,mueller2017correlating,francica2017daemonic,vitagliano2018trade,PhysRevLett.121.120602} in thermodynamic transformations in different scenarios. In general, the creation of correlations is associated to some energetic cost
and strategies to optimally extract work from them have been put forward \cite{vitagliano2018trade,bruschi2015thermodynamics,huber2015thermodynamic,PhysRevLett.121.120602,perarnau2015extractable,francica2017daemonic}. 
On the other hand, it has been demonstrated that in the single-shot regime, by allowing auxiliary correlated catalytic systems~\cite{Lostaglio2015}
or correlations with catalytic systems~\cite{mueller2017correlating}, it is possible to enlarge the set of achievable transformations.

\begin{figure}[b]
  \centering
  \includegraphics[width=0.9\linewidth]{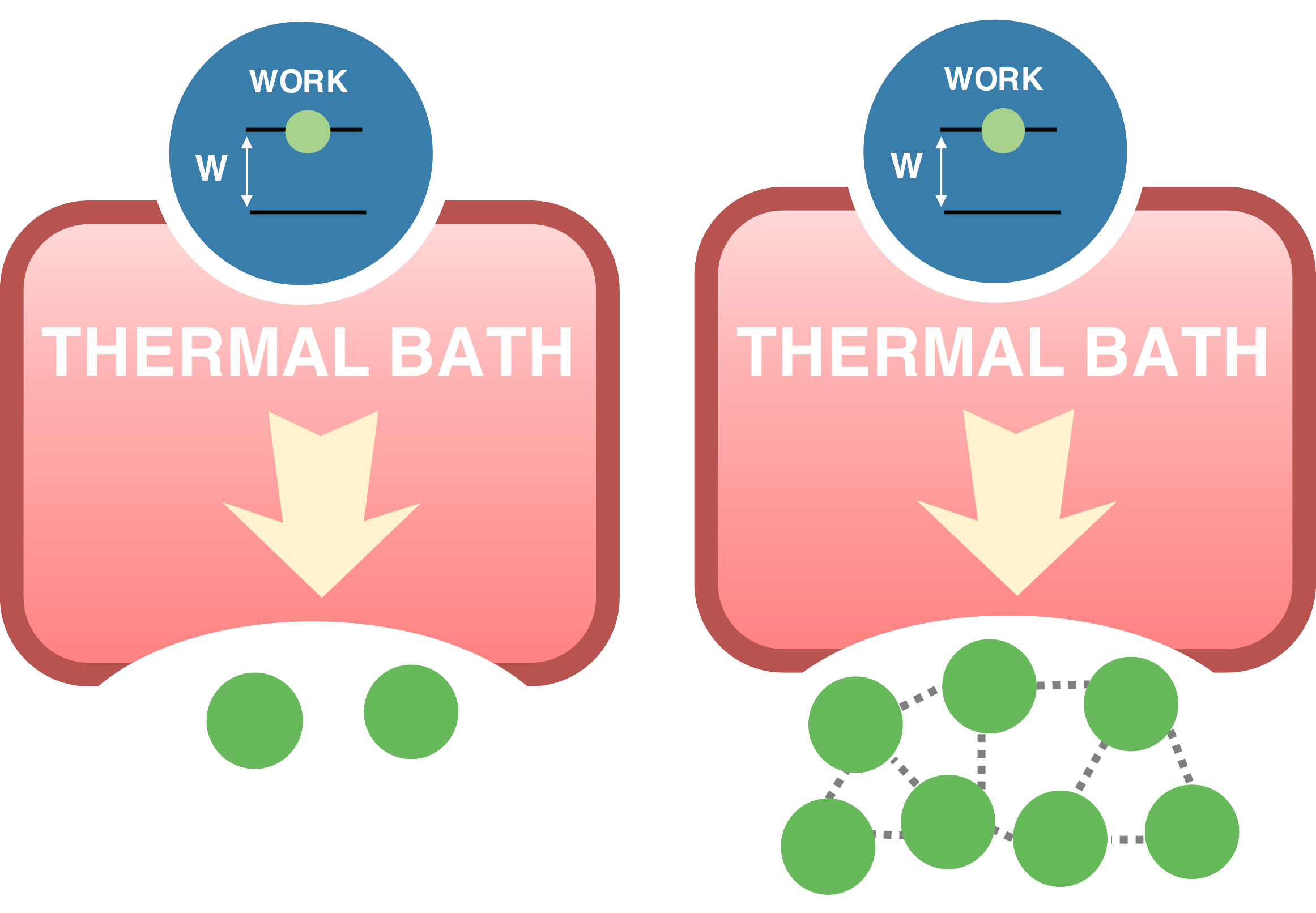}
  \caption{ \label{fig:introduccion} Correlations as a resource. In the single-shot regime, the creation of $N$
    correlated copies of a given state has a smaller work cost than creating $N$ independent copies. 
   In the asymptotic limit, this energetic cost per copy converges to the standard nonequilibrium free energy difference.}
\end{figure}

In this article, we study how inner correlations can affect certain fundamental processes that take place in contact with a thermal reservoir. 
In particular, we consider the processes of state formation and work extraction in the single-shot regime. That is, a Gibbs state is transformed into 
some out-of-equilibrium state using deterministic work, and deterministic 
work is drawn from an inverse transformation. 
We start our analysis by concentrating on the work of formation of a finite set of locally identical quantum systems.
We find that by allowing correlations in the final state this energetic cost can be reduced (Fig.~\ref{fig:introduccion}).
This is in strike contrast with standard scenario where arbitrary large fluctuations of work are allowed, and the 
creation of correlations requires some extra energy.
While for uncorrelated copies most of these processes are shown
to be irreversible, here we show that the degree of reversibility in the correlated scenario increases with the number of copies.
Then, we show that in the asymptotic limit the optimal collective process can be accomplished with correlations per particle that are 
vanishing small, and the work of formation per particle equals the free energy difference.
Finally, we generalize these results for an arbitrary set of local systems.


\section{Results}


\noindent\textbf{Overview.}
In standard thermodynamics, the transformations between states that occur in contact with a thermal reservoir
are governed by the Helmholtz free energy
	\begin{equation}	
    F(\rho)=\mean{E(\rho)} - k_\nsm{B} T S(\rho),
	\end{equation}
 where $\mean{E(\rho)}$ is the mean
energy of the system in state $\rho$, $S(\rho)$ is the entropy, $k_\nsm{B}$ is Boltzmann's constant, and $T$ is the temperature of 
the thermal reservoir. Transitions between states are allowed provided that the free energy of the final state is lower than at the beginning. In fact, the difference in free energy is equal to the amount of average work that can be extracted during the process,
and is also equal to the work that should be invested in the reverse process. This formulation was developed
for macroscopic systems where, due to the large number of particles, energy fluctuations become negligible. 
On the other hand, if one wishes to understand thermodynamic transformations of a small number of non-equilibrium systems,
the size of these fluctuations become important as they could be of the order of the value of work.
Recently, an approach that addresses thermodynamic transformations in this regime has been developed, and 
conditions on state transitions have been identified \cite{janzing2000thermodynamic,brandao2013resource,Horodecki2013,Aberg2013}. 
Below, we briefly introduce the formalism to study thermodynamic transformations in the single-shot regime, and review the main results relevant to this work.

At the core of the theory is the identification of a set of allowed operations, which model
the most general transformation in this framework \cite{Horodecki2013,brandao2013resource,Brandao2015}. 
Let us consider a system with Hamiltonian $H_\nsm{S}$ and an arbitrary thermal reservoir $\nsm{R}$ in a Gibbs state 
$\tau_\nsm{R} = e^{-\beta H_\nsm{R}} / \tr[{e^{-\beta H_\nsm{R}}}]$ 
with Hamiltonian $H_\nsm{R}$, and $\beta = 1/k_\nsm{B}T$. System and reservoir are allowed to interact via a unitary
evolution $U$ that preserves the total energy $[U, H_\nsm{S}+H_\nsm{R}]=0$, and then it is possible to perform a 
partial trace over $\nsm{S}$ and $\nsm{R}$. Given the system in an initial state, the allowed transformations 
are called \deit{thermal operations} and they define a set of reachable states.
In contrast to other frameworks, where just the conservation of the mean energy
is imposed \cite{skrzypczyk2014work}, these conditions give a strong conservation
of energy (first law of thermodynamics). Thus, given two states $\rho$ and $\sigma$, we
say that $\sigma$ can be reached from $\rho$, $\rho\rightarrow \sigma$, if there exists a thermal operation that implements such transformation.
A necessary condition for thermodynamic state transitions is 
called \deit{thermo-majorization} \cite{Horodecki2013}, that is 
sufficient for diagonal states, i.e. ${[\rho, H_\nsm{S}] = 0}$, which is also the case we will consider here.
Although the thermal operations appear potentially very
complex, since they allow any energy conserving interaction between state and bath, it
has been shown that they can always be achieved as sequences of elemental operations
that have a simple form and physical interpretation \cite{PhysRevX.8.041049}.

More generally, one can consider transformations that also allow the presence of an additional system that acts as a catalyst of the transformation,
 and is returned in the same state. These type of transformations enlarge the set of reachable states, and the necessary and sufficient conditions for diagonal states can be written as an infinite set of
inequalities~\cite{Brandao2015}. 
In this case, a transformation from an initial state $\rho$ to a final state $\sigma$ 
can be done provided ${F_\alpha(\rho) \geq F_\alpha(\sigma)}$ for all $\alpha \in \mathbb{R}$, where $F_\alpha$ are 
the $\alpha$-free energies defined in terms of R\'enyi divergences $D_\alpha(\rho\|\tau_\nsm{S})$ as
${F_\alpha ( \rho ) =k_\nsm{B}TD_\alpha(\rho \|\tau_\nsm{S})-k_\nsm{B} T \log Z_\nsm{S}}$, with $\tau_\nsm{S}=e^{-\beta H_\nsm{S}}/Z_\nsm{S}$ the thermal state
of the system. Thus, this is the family of inequalities that govern thermodynamic transformations in this regime~\cite{Brandao2015}. The standard second law is contained as a particular case for $\alpha = 1$.
 
In the single-shot scenario, the notion of deterministic work can be considered by introducing an auxiliary two-level 
system $\nsm{W}$ with Hamiltonian $H_\nsm{W}= W\dyad{W}_\nsm{W}$, called work qubit or wit, that acts as 
a battery which can store or inject energy into the system~\cite{Horodecki2013}.
In particular, we will be interested in the energetic cost of obtaining a state $\rho$ out from a thermal state.
This work cost can be evaluated by studying the following transformation:
\begin{equation}
  \tau_\nsm{S} \otimes \dyad{W}_\nsm{W} \rightarrow \rho \otimes \dyad{0}_\nsm{W}.
\end{equation}
The smallest possible value of such $\nsm{W}$ is defined as the work of formation~\cite{Horodecki2013}, 
and gives the minimum amount of deterministic work required in the transformation.
For diagonal states it is given by
\begin{equation}
   \Wform(\rho) = k_\nsm{B} T D_\infty(\rho \| \tau_\nsm{S}),
 \end{equation}
which is also equal to 
$k_\nsm{B} T \log \max_{E,g} \left \{ \lambda_{E,g}\, e^{\beta E} Z_\nsm{S}\right \}$,
where $\lambda_{E,g}$ are the eigenvalues of $\rho$, $g_\nsm{S}(E)$ the degeneracy, and $g=1,..,g_\nsm{S}(E)$.
Notice that the work of formation is in general greater than the free energy difference.
Similarly, one can define the extractable work as the maximum work that
can be stored in the work qubit starting from a state $\rho$, and its expression for diagonal states
is given by~\cite{Horodecki2013}:
\begin{equation}
  \WE(\rho) = k_\nsm{B} T D_0(\rho \| \tau_\nsm{S}).
\end{equation}
The addition of a catalyst to the process of state creation or work extraction does not modify these values of work.
Finally, let us mention an important feature of this theory that is also relevant to our work:
in general the extractable work is smaller than
the work of formation, thus there is an inherent irreversibility in thermodynamic 
transformations in this regime~\cite{Horodecki2013,Brandao2015}. 
However, when correlated catalysts are allowed, 
the transformations become ruled by just the usual free energy difference \cite{Lostaglio2015,mueller2017correlating}.
\\

\noindent\textbf{Work of formation of correlated copies.}
Let us consider a situation where a finite set of particles is prepared in the same reduced diagonal state $\rho$.
This could be for instance the first step of a given task. 
What is the minimum work cost of producing such $N$-partite ensemble if one is able to interact with a thermal reservoir? 
There are many multipartite states compatible with this situation, 
since it is only defined by some reduced state and number of particles, but these states have a different work cost. 
If we allow arbitrary large fluctuations of work \cite{Richens2016}, creating a correlated state $\rho^{(N)}$ out of a thermal one is useless. 
The average work cost associated to correlated copies of $N$ systems with 
Hamiltonian $H_\nsm{S}$ is given by the standard non-equilibrium free energy 
difference ${\mean{W} \equiv F(\rho^{(N)})-F(\tau_\nsm{S}^{\otimes N})}$ which can be expressed as
\begin{equation}
  \mean{W}= N \Delta F(\rho) + k_\nsm{B} T\;\mathcal I( \rho^{(N)}),
\end{equation}
where $\Delta F(\rho)=F(\rho)-F(\tau_\nsm{S})$, and 
\begin{equation}
\mathcal I( \rho^{(N)})\equiv D_1 \left( \rho^{(N)} \middle\| \rho^{\otimes N} \right)
\end{equation}
is a measure of the total correlations \cite{watanabe1960information}, with $D_1 (\, \cdot \, \| \, \cdot \,)$ the relative entropy.
This average work cost has two components: the energy required to obtain $N$ uncorrelated copies $N \Delta F(\rho)$ and the energy associated
to the correlations which is also positive. Therefore the above expression tells us that correlations are costly, 
if unbounded fluctuations of work are allowed, the work cost of this task cannot be reduced by creating correlations between subsystems. 
In what follows, we will show that in the single-shot scenario a collective action provides an advantage, and
in fact the minimum work cost of this task is achieved with correlated copies.

Let us consider $N$ identical $D$-dimensional quantum systems $\nsm{S}$ with Hamiltonian $H_\nsm{S}$.
Given a reduced state $\rho=\sum_{d=1}^{D} p_d\, \dyad{E_d}$, we are interested in studying the following transformation: 
\begin{equation}
  \tau_\nsm{S}^{\otimes N} \otimes \dyad{W} \rightarrow \rho^{(N)} \otimes \dyad{0},
\end{equation}
where with an amount of deterministic work $\nsm{W}$ a multipartite state $\rho^{(N)}$ is created, subject to the
local condition
\begin{equation}
  \label{eq:condition_partial_trace}
  \tr_{-i}(\, \rho^{(N)}) = \rho \quad \forall\, i = 1, 2, \ldots N,
\end{equation}
with $\tr_{-i}(\cdot)$ the partial trace over all the systems except
the $i$-th subsystem. Notice that we are considering \deit{exact transformations} for every $N$
and, as said before, we are not allowing fluctuations in the values of work~\cite{Richens2016}.
Let us call $\mathcal{C}(\rho, N)$ the set of all the diagonal states which
satisfy the partial trace condition of Eq.~\eqref{eq:condition_partial_trace}.
We can now define the \cwork $\Wmin(\rho, N)$ as
the minimum work cost of this transformation over all the states in $\mathcal{C}(\rho, N)$:
\begin{equation}
  \Wmin(\rho , N) = \min_{\rho^{(N)} \in\, \mathcal{C}(\rho , N)} \Wform(\,\rho^{(N)}) .
  \label{eq:wmin}
\end{equation}
In what follows we will show that it is possible to find this minimum work cost and characterize a set
of states that achieve this bound.

In order to carry out the minimization, first notice that the \cwork is always minimized by a state
$\optimal$ that is maximally mixed in each populated subspace of energy (see Supplementary Note 1 for details). 
Thus, one can reduce the set $\mathcal{C}(\rho, N)$, where the minimization is done, to a smaller subset of states.
These states are such that $\lambda_{E,g} = p_E / g_N(E)\equiv \lambda_E$ , where $p_E$ is the occupation of the subspace of energy $E \in \mathcal{E}_N$,
$\mathcal{E}_N$ is the spectrum of the $N$-partite system
and $g_N (\cdot)$ is the degeneracy. 
Therefore, each element of the subset is determined by just specifying the corresponding distribution $\lambda_E$.
Notably, the minimization in Eq.~\eqref{eq:wmin} can be written as an
optimization problem subject to linear constraints:
\begin{eqnarray} \label{eq:LP}
  \min_{\{\lambda_E\}_{E \in \mathcal{E}_N}} k_\nsm{B}T
    \log \left[ \max_E \right.  \left. \left\{ \lambda_E e^{\beta E}  Z_\nsm{S}^N  \right\} \right]  \\
 \text{s.t.}  \quad  \sum_{E \in \mathcal{E}_N} g_{N-1}(E-E_d)\lambda_E = p_d \quad \forall \,d =& 1,\ldots,D \nonumber\\
  \lambda_E \geq 0  \quad \forall\, E \in \mathcal{E}_N. \nonumber
  \end{eqnarray}
Moreover, this system of equations can easily be transformed into a linear
optimization problem \cite{luenberger1984linear}.
Both constraints define a bounded and non-empty convex set, and therefore there
exists at least one optimal feasible solution. Since the optimization problem is linear
there is an efficient algorithm, known as the \deit{simplex algorithm}, that allows to solve the problem numerically.
Furthermore,  we will show how to fully characterize $\Wmin(\rho , N)$ and the energy distribution $\lambda_E$ that solves the minimization problem
for every local state $\rho$ and number of copies $N$.
\\


\begin{figure*}[t]
  \centering
  \includegraphics{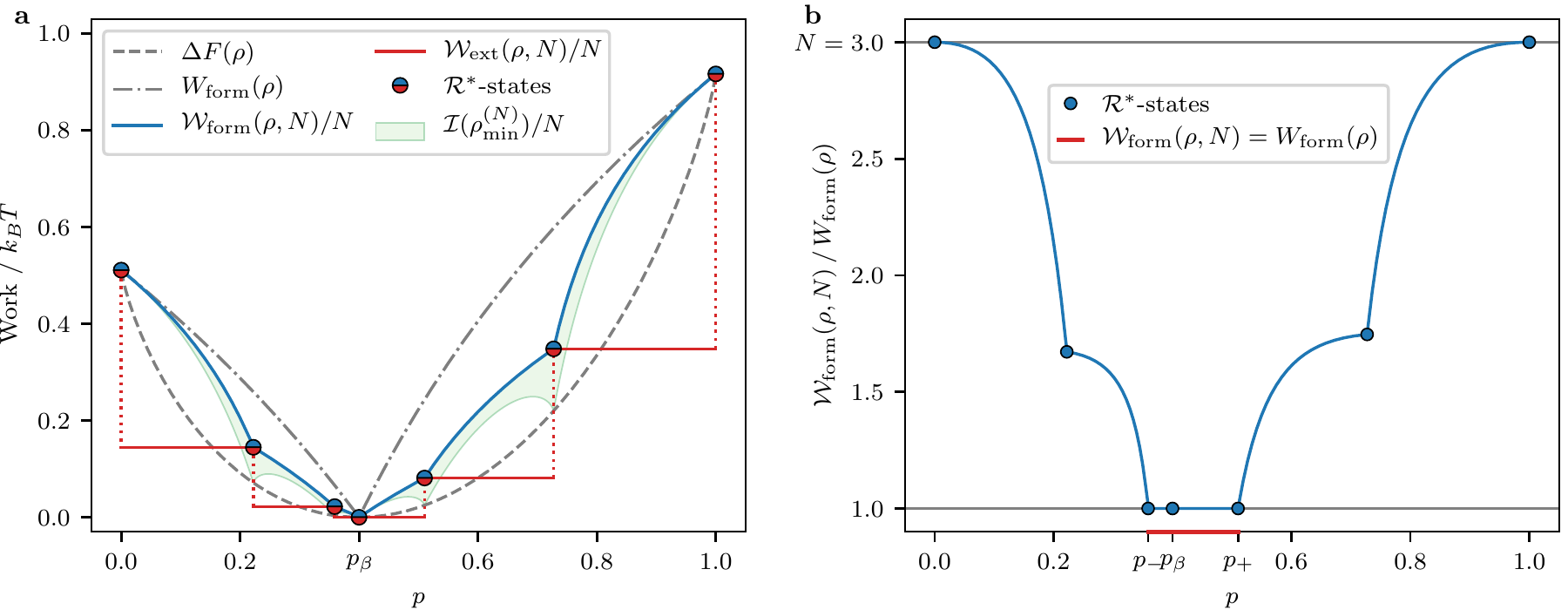}
  \caption{ \label{fig:trabajos}
    Minimum work of formation of three correlated copies for the different local qubit states.
    (a) Different measures of work as a function of the local qubit state, ${\rho = (1-p)\dyad{0} + p \dyad{1}}$, parametrized by the excited state probability $p$.
    For one copy, the work of formation $\Wform(\rho)$ (gray dash-dotted lines) is bigger
    than the standard free energy difference $\Delta F$ (gray dashed lines).
    On the other hand, the minimum work of formation of correlated copies (\cworkns) per copy, $\Wmin(\rho, N)/N$ (blue solid line),
    is smaller than or equal to $\Wform(\rho)$ but still greater than $\Delta F (\rho)$.
    $\Wext(\rho,N)/N$ (red solid line) is the extractable work per copy of the optimal correlated state $\optimal$.
    While in general all the correlated states are irreversible, $\Wext(\rho,N) < \Wmin(\rho, N)$, there are states (dots)
    whose $\optimal$ satisfies $\Wext(\rho,N) =\Wmin(\rho, N)$, this is the set
    of $\criticos$ (reversible optimal states). The green region represents the total correlations per copy $\mathcal I(\optimal)/N$ present in each multipartite state. 
    (b) Ratio between the \cwork and the work of formation of a single copy. Creating the correlated state $\optimal$ costs less work than $N$ uncorrelated copies.
    There are some extreme cases (red) around the thermal state $p_\beta$ where the work of formation of a single copy equals the work of formation of $N$ correlated copies.}
\end{figure*}

For simplicity, we will first present our results for the particular case where each subsystem has dimension $D=2$. 
However, our findings can be generalized to subsystems of arbitrary dimension, although the mathematical treatment is more involved.
Without loss of generality we will consider ${H_\nsm{S} = E_0 \dyad{1}}$ as the Hamiltonian of each subsystem, and a general diagonal local state ${\rho = (1-p)\dyad{0} + p \dyad{1}}$.
In this way, the local thermal Gibbs state is defined as the state with $p=p_\beta$, where $p_\beta = (1+e^{\beta E_0})^{-1}$, and
partition function ${Z_\nsm{S}=\tr(e^{-\beta H_\nsm{S}})}$.
Our first result shows the analytical solution to the optimization problem of Eq.~\eqref{eq:LP}.

\begin{Theorem} \label{theo:minimal}
  Given an integer $N$ and a state $\rho$ which satisfies $[\rho, H_\nsm{S}]=0$,
  there exists a subset of energies ${\mathcal{E}_N^{\rho} \subseteq\mathcal{E}_N}$, 
  a constant ${s}\in(0,1]$, and at most a single energy $\varepsilon \in \mathcal{E}_N$ such that 
  the state $\optimal$ is defined by the distribution:
  \begin{equation}
    \label{eq:solutionEnergy}
    \lambda_{E} = \left \{ \begin{array}{l l}
     \; \;\;\frac{e^{-\beta E}}{\gamma} & \text{if } E \in \mathcal{E}_N^{\rho} \vspace{0.4em} \\
      s \;\frac{e^{-\beta\varepsilon}}{\gamma}\, & \text{if } E = \varepsilon \vspace{0.4em} \\
      0 & \text{otherwise } \\
    \end{array} \right. ,
  \end{equation}
  with $\gamma$ a normalization constant. Furthermore, the work of formation and the extractable work of the optimal state $\optimal$ are given by
\begin{align}
  \Wmin(\rho, N) &= k_\nsm{B} T\, \log\left[ \frac{Z_\nsm{S}^N}{\gamma}\right], \label{eq:Wopt}\\
  \Wext(\rho,N) &= k_\nsm{B} T\, \log\left[ \frac{Z_\nsm{S}^N}{Z}\right], \label{eq:WoptE}
\end{align}
respectively, where $Z$ is the partition function of a 
system in a thermal state at temperature $T$ with spectrum given by the set ${\mathcal{E}_N^{\rho}\cup \{\varepsilon}\}$,
and $\gamma=Z- (1- s) \, g_N(\varepsilon) e^{-\beta \varepsilon}$. 
\end{Theorem}
  \begin{proof}
  See Supplementary Note 2.
  \end{proof}

The structure of the optimal states is simple: except for the occupation of a single level with energy $\varepsilon$,
$\optimal$ is a Gibbs thermal state over a reduced support of energies $\mathcal{E}_N^{\rho}$ 
which depends upon the local state $\rho$ and the number of copies $N$. 
The optimal states do have correlations that are the result of removing the population of some energy levels from the thermal state. 
Notice that a typical approximation that is usually done when one deals with large number of identical systems 
is similar to what is obtained in Eq.~\eqref{eq:solutionEnergy}, i.e., discard 
tails of the energy distribution \cite{brandao2013resource}.

In this way, Eq.~\eqref{eq:Wopt} gives the optimal work cost for creating a set of $N$ particles in a given local state.
Example calculations of the \cwork per copy $\Wmin/N$, the work of formation of a single copy
$\Wform$, and the amount of correlations in the optimal state for $N=3$ are shown in Fig.~\ref{fig:trabajos}-a). The \cwork per copy lies below the work of formation of a single copy. The 
difference between these two curves is precisely the energy per copy that is saved in the process of formation due to the collective action.  Fig.~\ref{fig:trabajos}-b) further stresses the difference between the \cwork per copy and the work of formation of a single copy. 
In fact, one can notice that there exist extreme cases, near the thermal state, where this ratio is minimal: $\Wmin(\rho, N) = \Wform(\rho)$.
These states are such that the work of formation of a single copy is equal to the amount of work one should invest to obtain $N$ correlated copies, but on the other hand one 
cannot extract any deterministic work from them (for a more detailed explanation of these states, see Supplementary Note 4). 

We have seen that the work of formation can be reduced if one acts collectively and creates correlations in the final state. This property appears 
when work is not allowed to fluctuate and thus the work of formation is greater than the free energy difference. 
However, not always the presence of correlations will help in the process. 
Notably, there is an upper bound on the amount of correlations that can be built up while reducing the work of formation:
\begin{equation}
\frac{1}{N}\mathcal I(\rho^{(N)}) \leq \beta\, \delta Q,
\label{eq:Imax}
\end{equation}
where $\delta Q \equiv \Wform(\rho)-\Delta F(\rho)$ is called the dissipated work, the difference 
between the deterministic work of formation of a single copy with the free energy difference 
(see Supplementary Note 3 for a proof of the bound).
Thus, correlations greater than $N\beta \delta Q$ are costly, since collective operations cannot outperform the single copy creation.
\newline

\noindent\textbf{Reversibility.}
A key result in the single-shot regime is the appearance of an intrinsic
irreversibility: the extractable work is in general smaller than the work of formation. Thus, in general,
one cannot extract the same amount of energy invested in the process of creation. 
However, it can be easily seen that there are families of states whose work of
formation and extractable work coincide, and in this sense these states are reversible.
Theorem 1 shows that, in fact, reversibility appears naturally in our framework.
The states we define are such that in general the \cwork is greater than the extractable work,
and the difference between these values is the irreversible work:
	\begin{equation}
\Wirr(\rho, N)=k_\nsm{B} T \log\left[{Z}/{\gamma}\right]. 
	\end{equation}
For $ s\approx 1$ the irreversible work is ${\Wirr(\rho, N)\approx(1- s) g_N(\varepsilon) e^{-\beta \varepsilon}/Z}$.
Remarkably, there is a subset of reduced states $\rho_k^*$ 
for which their corresponding $\optimal$ is a thermal state over the reduced support $\mathcal{E}_N^{\rho} \cup \{\varepsilon\}$ (i.e. $s = 1$). These states
are such that their work of formation is equal to the extractable work, and in this sense they are strictly \deit{reversible}, i.e. $\Wirr(\rho_k^*, N)=0$. On the other hand, the irreversibility increases as $ s\rightarrow 0$.

We call the set of local states whose $\optimal$ are reversible $\criticos$.
This set is composed by the states that match the break points in the curves of Fig.~\ref{fig:trabajos}.
 In Supplementary Note~4 it is shown that there are $2N+1$ of such states $\rho^*_k = (1-p^*_k)\dyad{0} + p_k^* \dyad{1}$ with $k=1,2,\ldots, 2N+1$ (see Fig.~\ref{fig:trabajos}-a)). 
Furthermore, for these states the work (either of formation or extractable) can be expressed as:
\begin{equation} 
\mathcal W(\rho_k^*, N) = N \Delta F(\rho_k^*) + k_\nsm{B} T \,\mathcal I_k.
\end{equation}
where $\mathcal I_k$ is the amount of total correlations present in the optimal state.
Thus, work is the sum of two contributions:
the classical value of work, given by the free-energy difference, plus the energy 
associated to the creation of correlations, but still ${\Wmin(\rho_k^*, N) \leq N\,\Wform(\rho_k^*)}$. 
This shows that collective operations allow us to perform 
reversible transformations using deterministic work.
Moreover, in the optimal process it is the energy of the correlations that fills the gap between the standard work of formation of independent copies and the $c$-work of formation
(see Fig~\ref{fig:trabajos}-a)).
Furthermore, as we will see below, these states have other interesting properties that will allow us to recover standard results from
thermodynamics in the large $N$ limit.
\\

\begin{figure*}[t]
  \centering
  \includegraphics{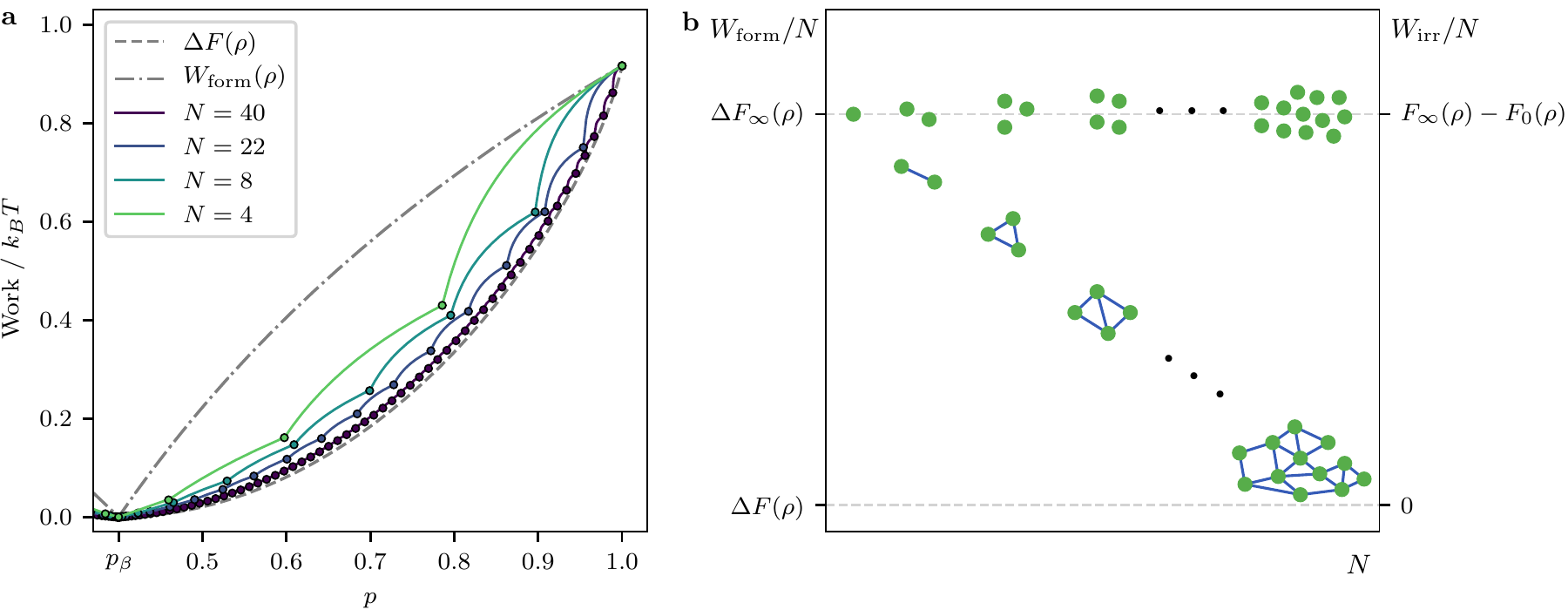}
  \caption{
    Asymptotic behavior of the minimum work of formation of correlated copies as the number of copies increases.
  (a) The minimum work of formation of correlated copies (\cworkns) per copy $\Wmin/N$ is plotted in color for
    different values of $N$ (solid lines) as a function of the local reduced
    state, $\rho = (1-p)\dyad{0} + p\dyad{1}$, parametrized by the excited
    state probability $p$. As $N$ increases $\Wmin/N$ approaches the standard
    free energy difference $\Delta F$ (dashed gray line) for every state and
    the number of reversible states ($\criticos$) (represented by dots) increases linearly with $N$. In addition, for irreversible states the degree of irreversibility per copy, measured by the difference between 
 the work of formation (solid colored line) and the extractable work (nearest point to the left), decreases with $N$.
 (b) Illustration of the results obtained for correlated copies. As $N$ increases the work of formation per copy approaches the free energy difference
and furthermore the irreversible 
 work per copy goes to zero; thus recovering reversibility. } \label{fig:tlimit} 
\end{figure*}


\noindent\textbf{Finite-$N$ behavior and thermodynamic limit.}
We have established that the \cwork represents the minimum amount
of energy that is necessary to produce $N$ correlated copies of a state $\rho$ in a deterministic process. 
The natural question that follows is how these results behave as the number
of copies increases. Fig.~\ref{fig:tlimit}-a) illustrates the behavior of the \cwork per copy for a few values of $N$. 
There, it is shown that the \cwork per copy approaches the free energy difference as $N$ increases.
In addition, the set of $\criticos$ increases linearly with $N$.

Further insight can be gained if one considers the density of reversible states. 
In Supplementary Note 4 we show that for any local state $\rho$ and $\epsilon > 0$, there exists a
number of copies $N = \mathcal{O}(1/\epsilon)$ and an $\critico$ with density matrix $\rho^*(\epsilon)$ that 
satisfies ${\| \rho - \rho^*(\epsilon) \| < \epsilon}$.
This means that the set of $\criticos$ is dense in the space of states with the local constraint Eq~\eqref{eq:condition_partial_trace}.
Moreover, since the irreversible work per copy goes to zero with $N$, for a large number of copies the process of formation is almost reversible.
Thus, in the thermodynamic limit, we recover the standard results from thermodynamics:

\begin{Theorem}
  Let $\rho$ be any diagonal local state of a system $\nsm{S}$. Then
  \begin{equation}
    \frac{\mathcal W(\rho, N)}{N} \xrightarrow{N \rightarrow \infty} \Delta F(\rho)
  \end{equation}
where $\mathcal W$ refers to either the \cwork or the extractable work of the optimal states,
  $\Delta F(\rho)\equiv F(\rho)-F(\tau_\nsm{S})$ is the standard nonequilbirium free energy difference, and the rate of convergence is $\mathcal{O}(\log N / N)$.
\end{Theorem}
  \begin{proof}
  See Supplementary Note 5.
  \end{proof}

As it is illustrated in Fig.~\ref{fig:tlimit}-b) by increasing the number of correlated copies we approach standard thermodynamics.
Additionally, for large $N$ the total correlations in the optimal state are of order: 
\begin{equation}
  \mathcal I (\optimal) \sim \mathcal{O}(\log N),
\end{equation}
meaning that the amount of correlations per particle $\mathcal I(\optimal)/N$ is negligible in the thermodynamic limit. 
These results establish that asymptotically the work per
copy required to form $N$ correlated states is exactly what we expect when unbounded fluctuations are allowed, with an amount of correlations that increases sublinearly with the number of copies. Previous analysis of the thermodynamic limit 
\cite{brandao2013resource,Chubb2018beyondthermodynamic}, considered interconversion rates of product states 
in the limit of large number of particles using approximate transformations. 
Here, we consider locally exact transformations,  
and find the solution that ranges from small number of systems to arbitrary large ones.
In this way, by obtaining the exact minimum work cost for every $N$ we could evaluate 
the deviation from standard results in every instance. 
More importantly, we have shown that in this approach the creation of correlations 
is the physical mechanism that leads to the emergence of the standard scenario.
\\


\noindent\textbf{Generalizations.}
These results were presented using the simplest example given by local systems of dimension $D=2$. In fact, 
more complex systems can be considered by increasing $D$, and we can show that these ideas also hold for arbitrary local dimension $D$ (see Supplementary Note 6). The main difference with respect to the case $D=2$ is that instead of having a single energy $\varepsilon$ (see Theorem 1),
each optimal state $\optimal$ is obtained by considering a set of energies $\varepsilon_i$ and 
parameters $ s_i \in (0,1]$, with $i=1,...,D-1$. The results concerning the thermodynamic limit have the same form. 

Up to now, we have focused on the situation where all subsystems have the same Hamiltonian and same reduced state.
 A natural extension of our findings is to consider a non-symmetric case, where each subsystem is different. There, one can also show that correlations reduce the work of formation, and that the optimal state has a thermal-like distribution similar to the one in Theorem 1 (see Supplementary Note 7). Furthermore, by allowing correlations it is possible to recover standard results in the thermodynamic limit for a general configuration in our framework, that is for a set of different diagonal states with vector probabilities $p_i$ and Hamiltonians $H_i$ taken from an arbitrary distribution $\mathcal D$.

	\begin{Theorem}
	Let $(p^{(1)}, E^{(1)}), (p^{(2)}, E^{(2)}), \ldots , (p^{(N)}, E^{(N)}) \in {\mathbb R}_{\ge 0}^{2D}$ an i.i.d sample with 
	arbitrary distribution $\mathcal D$ and $\mathcal W_N$ the \cwork of 
	a system with diagonal reduced states $\rho_i$ defined by the probability vector $p^{(i)}$ and Hamiltonian 
	with energies $E^{(i)}$. Then, 
	\begin{equation}
    \frac{\mathcal W_N}{N} \xrightarrow{N \rightarrow \infty} \mean{ \Delta F\, }_\mathcal D,
	\end{equation}
	where the mean in $\Delta F$ is with respect to $\mathcal D$ and the convergence is almost surely.
	\end{Theorem}
	\begin{proof}
	See Supplementary Note 7.
	\end{proof}

In this case the \cwork $\mathcal W_N$ can be thought as a random variable since the state of the $N$ subsystems is chosen randomly following the distribution $\mathcal D$. 
For instance, if $\mathcal D$ has density $f$ taking values in $\Omega$,
 then 
	\begin{equation}
	\mean{ \Delta F }_{\mathcal D} = \int_\Omega \Delta F \left(p,E \right) f(p,E) \,dp \,dE.
	\end{equation}
When dealing with copies the distribution is defined by ${f(p,H)=\delta(p-\tilde p)\, \delta(E-\tilde E)}$, where $\delta(\cdot)$ is the Dirac delta distribution. Thus, our approach can be directly extended to more general settings including the asymmetric case.


\section{Discussion}


We have presented a framework to study how the presence of correlations 
affects thermodynamic processes taking place in the single-shot regime.
By first considering the formation of locally equivalent states,
we show that the creation of correlated systems provides an advantage, since in this case the energetic
cost of the process is lower than in the uncorrelated scenario. This is a feature that appears when fluctuations 
of work are constrained.  Although we  focused most 
of our analysis on the creation correlated of copies, we have shown that the same ideas could be extended to more general scenarios.
Here, we have analyzed the creation of states that are diagonal in the energy eigenbasis. Creation of states with coherence in this scenario is strictly impossible, a source of coherences is required \cite{brandao2013resource,Lostaglio2015a}. 
We think that these ideas could also be extended to this situation, as the amount of coherences can be 
reduced when one acts collectively. If independent copies require an amount of coherences
of order $\mathcal O(N)$, a collective action reduces this requirement to $\mathcal O(\sqrt N)$ \cite{brandao2013resource,Lostaglio2015a}.

The description we provide is compatible with standard results of thermodynamics 
in the large $N$ limit. In fact, we have shown that this mechanism leads to the emergence of reversibility when the minimum work
cost is considered. Unlike previous approaches,
here we consider that the final state is correlated but all transformations and work
extraction are exact and deterministic. Interestingly, we have also shown that an amount of correlations per 
copy that is vanishing small is sufficient to recover standard results in the large $N$ limit. 
Therefore, we can identify a physical mechanism, related with the creation of correlations, that allows 
to continuously approach standard results in the thermodynamic limit. 
Furthermore, we have also shown that classical results can also be recovered in the large $N$ limit with more general settings.
We expect our work sheds light on the role of correlations in thermodynamic transformations of microscopic systems and its
connection with the emergence of standard results.
\\

\noindent
\textbf{Data availability.}
 The data that support the findings of this study are available from the corresponding authors upon request..
\\

\noindent
\textbf{Acknowledgments.}
We thank L. Masanes for useful comments, also P. Groisman and G. Acosta for discussions. 
FC and AJR acknowledge support from CONICET, UBACyT and ANPCyT.
\\

\noindent
\textbf{Author contributions.}
F. S., F. C. and A.J.R. contributed to all aspects of this work.
\\

\noindent
\textbf{Competing interests.}
The authors declare no competing interests.
\\


\def\bibsection{\section*{\refname}}


\clearpage

\onecolumngrid

\appendix

\renewcommand{\appendixname}{Supplementary Note}
\renewcommand{\thesection}{\arabic{section}}
\renewcommand{\thesubsection}{\Alph{subsection}}
\newtheorem{PropositionSM}{Proposition}
\newtheorem{CorollarySM}{Corollary}
\newtheorem{LemmaSM}{Lemma}
\newtheorem{TheoremSM}{Theorem}
\newtheorem{DefinitionSM}{Definition}

\newcommand{\figref}[1]{Fig.~\ref{#1}}



\section{\cwork}

In this work we focus our study on the process of formation of $N$ correlated
copies of a system $\system$ with non-trivial Hamiltonian $\HS$. More specifically, given a
state $\rho$ of the system with $[\rho, \HS]=0$ and a state $\rho^{(N)}$ which
satisfies
\begin{equation}
  \label{eq:condicionTrazaParcial}
  \tr_{-i}\left( \rho^{(N)} \right) = \rho \quad \forall i = 1,2,\ldots, N ,
\end{equation}
where $\tr_{-i} (\cdot)= \tr_{1,2,\ldots, i-1, i+1, \ldots,
N}(\cdot)$ is the partial trace over all subsystems except for the $i$-th one,  
we consider the transformation via thermal operations given by
$  \tauS^{\otimes N}\otimes\dyad{W} \Ttransform \rho^{(N)}\otimes\dyad{0} ,$
where $\tauS$ is the Gibbs state of S. The aim is to find states $\rho^{(N)}$
with minimum work of formation $W$. We can formalize this problem as a
minimization problem over a set of feasible solutions.

\begin{DefinitionSM}[feasible states]
	\label{defi:conjuntoFactible}
  Given a state $\rho$ of the system $\system$ and $N \in \mathbb{N}$, let
  $\mathcal{C}(\rho, N)$ be the set of states of the $N$ copies,
  $\system^{\otimes N}$, given by
  \begin{equation}
    \mathcal{C}(\rho, N) = \left\{ \rho^{(N)} : \tr_{-i}\left(\rho^{(N)}\right) =
    \rho \quad \forall i = 1, 2, \ldots, N \right\}.
  \end{equation}
\end{DefinitionSM}

Except for the product state $\rho^{\otimes N}$, the set $\mathcal{C}(\rho,
N)$ is formed by correlated states of the $N$ subsystems.
The \cwork~is then defined as the minimum amount of energy
that is required to create a state in $\mathcal{C}(\rho, N)$
from the Gibbs state of the total system.

\begin{DefinitionSM}[work of formation of correlated copies]
  Given a state $\rho$ of the system $\system$ and a number of copies $N$, we
  define the work of formation of correlated copies or just \cwork , denoted as $\Wmin(\rho, N)$, by
  \begin{equation}
    \Wmin(\rho, N) = \min_{\rho^{(N)} \in \mathcal{C}(\rho, N)}
    \Wform\left(\rho^{(N)}\right).
  \end{equation}
\end{DefinitionSM}

Let $D$ be the dimension of the system $\system$ and $E_1 \leq E_2 \leq \ldots
\leq E_D$ the eigenvalues of $\HS$. Let's assume that there are at least two different eigenvalues 
(a similar treatment can be considered for trivial Hamiltonians and the solution
is similar to the one that is obtained in the limit $\beta\rightarrow 0$). Then the states $\rho$ and $\rho^{(N)}$ can
be written as
\begin{gather}
  \rho = \sum_{d=1}^D p_d \dyad{E_d} , \\
    \rho^{(N)} = \sum_{\left(d_1,\dots,d_N\right)\in\left\{1,\dots,D\right\}^N} \lambda_{d_1, d_2, \ldots, d_N}
  \dyad{E_{d_1}, E_{d_2}, \ldots, E_{d_N}} ,
  \label{eq:rhoNgeneral}
\end{gather}
where $p_d$ are the eigenvalues of the state $\rho$ and $\lambda_{d_1, d_2,
\ldots, d_N}$ are the occupation probabilities of the state $\rho^{(N)}$ with
the $i$-th subsystem in the energy state $E_{d_i}$. Applying $\tr_{-j}(\cdot)$
in \eqref{eq:rhoNgeneral} we find that the constraint of
\eqref{eq:condicionTrazaParcial}  on the partial traces  will be satisfied if
and only if
\begin{equation}
  \label{eq:condition1_pd}
    p_d = \sum_{\left(d_1,\dots,d_{j-1},d_{j+1},\dots,d_N\right)\in\left\{1,\dots,D\right\}^{N-1}} \lambda_{d_1, \ldots, d_{j-1}, d,
  d_{j+1},  \ldots , d_N}	 \quad \forall\, d,j.
\end{equation}
On the other hand, the Gibbs state of the $N$ copies is given by
\begin{equation}
    \tauS^{\otimes N} = \sum_{\left(d_1,\dots,d_N\right)\in\left\{1,\dots,D\right\}^N} \frac{e^{-\beta \left(
  E_{d_1}+E_{d_2}+\ldots +E_{d_N} \right)}}{\ZSN} \dyad{E_{d_1}, E_{d_2},
  \ldots, E_{d_N}},
\end{equation}
where $\ZS = \sum_{i=1}^D e^{-\beta E_i}$ is the partition function of the
system $\system$ at inverse temperature $\beta = (\kBT)^{-1}$. Then, the work of
formation of $\rho^{(N)}$ is
\begin{equation}
  \label{eq:wFormacionN}
  \Wform(\rho^{(N)}) = \kBT \log \max_{d_1, \ldots , d_N} \Big \{
  \lambda_{d_1, \ldots , d_N}e^{\beta \left( E_{d_1}+\ldots +E_{d_N}
  \right)}\ZSN  \Big \}.
\end{equation}
Let $p_E$ be the occupation of energy level $E$ of the $N$-partite system,
which is given by
\begin{equation}
  p_E = \sum_{d_1, \ldots, d_N \,:\, \sum_{i=1}^N E_{d_i} = E} \lambda_{d_1,
  d_2, \ldots, d_N},
  \label{eq:pEdef}
\end{equation}
and let $g_N(E)$ be its degeneracy (i.e. the number of terms in the sum of
\eqref{eq:pEdef}). Then, it is clear from \eqref{eq:wFormacionN} that in order to minimize
the work of formation, the distribution of $p_E$ in the $g_N(E)$ states should be uniform:
\begin{equation}
  \lambda_{d_1, \ldots, d_N} = \frac{p_{\sum_{i=1}^N E_{d_i}}}
  {g_N\left(\sum_{i=1}^NE_{d_i}\right)} \equiv \lambda_{\sum_{i=1}^N E_{d_i}}.
\end{equation}
Given this symmetry, we introduce the notation $\lambda_E = \lambda_{d_1,
\ldots, d_N}$, with $E = \sum_{i=1}^N E_{d_i}$. Then, we can restrict the search for a state $\rho^{(N)}$ that minimizes $\Wmin$
to states of the following form
\begin{equation}
  \label{eq:minEsEnergia}
  \rho^{(N)} = \sum_E p_E \sum_{\psi \,:\, E(\psi)=E} \frac{1}{g_N(E)} \dyad{\psi}.
\end{equation}
This simplifies the problem of finding the state $\rho^{(N)}$ that minimizes the work
of formation. The state of \eqref{eq:minEsEnergia} is completely characterized by
its energy distribution among the set $\mathcal{E}_N$ of energies of the
$N$-partite system. Now the constraints of \eqref{eq:condition1_pd} on the reduced states
take a simpler form
\begin{align}
    p_d &= \sum_{\left(d_1,\dots,d_{j-1},d_{j+1},\dots,d_N\right)\in\left\{1,\dots,D\right\}^{N-1}} \lambda_{d_1, \ldots, d_{j-1}, d,
  d_{j+1}, \ldots, d_N} \nonumber \\
  &=  \sum_{E \in \mathcal{E}_N}
  g_{N-1}(E-E_d) \lambda_E \quad \forall d = 1,2, \ldots, D,
  \label{eq:conditionSimplificada}
\end{align}
where $\lambda_E$ is any of the eigenvalues of $\rho^{(N)}$
with energy $E$. Notice that each of the terms $g_{N-1}(E-E_d)\lambda_E$
represents the conditional probability that one of the copies have local energy
$E_d$ given that the total energy is $E$. The following proposition summarizes
the previous analysis.

\begin{PropositionSM}
  Given a local state $\rho$ and $N \in \mathbb{N}$, it holds
  \begin{equation}
    \min_{\rho^{(N)} \in \mathcal{C}(\rho, N)} \Wform(\rho^{(N)}) =
    \min_{\rho^{(N)} \in \mathcal{C}^*(\rho, N)} \Wform(\rho^{(N)}),
  \end{equation}
  where $\mathcal{C}^*(\rho, N)$ is the set of states of the form
  given by \eqref{eq:minEsEnergia} which satisfies \eqref{eq:conditionSimplificada}.
\end{PropositionSM}


\subsection{\cwork~as a linear program}

From the results of the previous section it is clear that the search for the
state $\rho^{(N)}$ with minimum work of formation is equivalent to solving the
following constrained minimization problem
\begin{align}
  \min_{\{\lambda_E\}_{E \in \mathcal{E}_N}}  \kBT \log \left[ \max_E \left\{
  \lambda_E e^{\beta E} \ZSN \right\} \right] \nonumber \\
  \text{s.t.} \qquad  \sum_{E \in \mathcal{E}_N} g_{N-1}(E-E_d)\lambda_E &= p_d
  \quad &\forall d = 1,2,\ldots, D \label{eq:main_optimization}\\
  \lambda_E & \geq 0  \quad &\forall E \in \mathcal{E}_N. \nonumber
\end{align}
To simplify the problem, we perform the change of variables 
\begin{equation}
  q_E = \lambda_E e^{\beta E}\ZSN \label{eq:change_variables}.
\end{equation}
We will also omit the term $\kBT$ (being a positive constant)
and drop the function $\log$ (given that $\log$ is a monotonically increasing function) obtaining
\begin{align}
  \min_{q} \quad \|q\|_\infty  \label{eq:minimization.posta} \\
  \text{s.t.} \qquad \sum_{E \in \mathcal{E}_N} \frac{g_{N-1}(E-E_d)}{\ZSN}
  e^{-\beta E} q_E &= p_d \quad &\forall d = 1, 2, \ldots, D
  \label{eq:constrain1.posta} \\
  q_E &\geq 0 \quad &\forall E \in \mathcal{E}_N,
  \label{eq:constrain2.posta}
\end{align}
where $\| \cdot \|_\infty$ is the infinity norm. Finally the minimization of the
infinity norm can be linearized by introducing an auxiliary variable $Q$ and
additional constraints
\begin{align}
  \min_{Q} Q \nonumber \\
  \text{s.t.} \qquad \sum_{E \in \mathcal{E}_N} \frac{g_{N-1}(E-E_d)}{\ZSN}
  e^{-\beta E} q_E &= p_d \quad &\forall d = 1, 2, \ldots, D \nonumber \\
  q_E &\geq 0 \quad &\forall E \in \mathcal{E}_N \\
  q_E &\leq Q \quad &\forall E \in \mathcal{E}_N \nonumber \\
  Q &\geq 0 . \nonumber
\end{align}

As mentioned in the main text, linear optimization problems have been
extensively studied and there exist numerical methods, such as the
\deit{simplex algorithm}, that allows one to solve this problem easily for
arbitrary energy distributions. In the following section we will focus on the
analytical solution for the qubit ($D = 2$) case, and then we will
generalize these results to arbitrary dimensions.


\section{Optimal solution for $D=2$}

When $D=2$ the Hamiltonian of the system can be considered, without loss of
generality, as $\HS = E_0 \dyad{1}$,
where $E_0$ is the energy of the excited state. Thus, every block-diagonal state of $\system$
can written as
$\rho = (1-p)\dyad{0} + p \dyad{1}$,
with $p \in [0,1]$ the probability of being in the excited state.

The advantage of considering $D=2$ is that there is a simple closed formula for the
degeneracy of the  $N$-partite energy levels in the set $\mathcal{E}_N = \{ mE_0, \text{with } m = 0, 1,  \ldots,
N\}$, that is given by
\begin{equation}
  g_N(mE_0) = \binom{N}{m}, \quad m=0,1,\ldots, N,
\end{equation}
where $\binom{N}{m}$ is the binomial coefficient.
Then, the constrains of \eqref{eq:constrain1.posta} and \eqref{eq:constrain2.posta}
can easily be written as:
\begin{align}
  \sum_{m=0}^N \frac{1}{\ZSN} \binom{N}{m} e^{-\beta mE_0} q_m &= \sum_{m=0}^N
  g_m q_m = 1 , \\
  \sum_{m=0}^N \frac{1}{\ZSN} \frac{m}{N} \binom{N}{m} e^{-\beta mE_0} q_m &=
  \sum_{m=0}^N \frac{m}{N}g_mq_m = p ,
\end{align}
where $g_m = \frac{1}{\ZSN} \binom{N}{m} e^{-\beta mE_0}$ and we have used the
fact that ${N - 1 \choose m - 1} = \frac{m}{N} {N \choose m}$.

The following Theorem characterizes the state $\optimal$ which minimizes the
work of formation.
\begin{TheoremSM}[Exact solution, $D=2$]
  Consider the optimization problem
  \begin{align}
    \min_{q \in \R^{N+1}}  \quad \| q \|_\infty \nonumber & \\
    \text{s.t.} \qquad \sum_{m=0}^N g_m q_m &= 1 \label{eq:constrain1} \\
                      \sum_{m=0}^N mg_mq_m &= Np \label{eq:constrain2} \\
                      q_m & \geq 0, \label{eq:constrain3}
  \end{align}
  where $p\in [0,1]$, $N\in \mathbb{N}$ and $g_m$ are non-negative real
  numbers. Then, there exists a unique solution and it is given by
  \begin{equation}
    q_m^* = \left \{ \begin{array}{l l}
      \gamma & \text{if } m \in U \\
      0 & \text{if } m \in L \\
      s\gamma & \text{if } m \in (U \cup L)^c  \\
    \end{array} \right.
    \label{eq:solutionD2}
  \end{equation}
  where $\gamma \in \mathbb R_{>0}$, $s \in [0,1]$ and $U$, $L$ are two different  intervals of indexes such that $\#((U \cup L)^c) \leq 1$ (i.e. there is at most only one element not in either $U$ or $L$). Moreover, if $(U \cup L)^c = \{ m^* \}$ then $U=\{ 0,1,\ldots, m^*-1\}$ and $L=\{ m^*+1, m^*+2, \ldots, N \}$, or
$U=\{ m^*+1, m^*+2, \ldots, N \}$ and $L=\{ 0,1,\ldots, m^*-1\}$; and $\gamma$ and $s$ are given by
  \begin{equation}
    \gamma = \frac{1}{sg_{m^*} + \sum_{m \in U} g_m} = \frac{m^* - Np}{\sum_{m \in U} (m^*-m)g_m} \quad , \quad
    s = \frac{1}{g_{m^*}}\frac{\sum_{m \in U}(Np-m)g_m}{m^* - Np}.
    \label{eq:Zys}
  \end{equation}
  On the other hand, if $(U \cup L)^c = \emptyset$ then $s = 0$ and
  \begin{equation}
  \gamma = \frac{1}{\sum_{m \in U} g_m}.
\end{equation}   
  \label{teo:Caracterizacion2D}
\end{TheoremSM}

\begin{proof}
Since $\|\cdot\|_\infty$ is a metric in $\R^{N+1}$ and the set of \eqref{eq:constrain1}, \eqref{eq:constrain2} and \eqref{eq:constrain3} define a
convex bounded set, there exists at least one solution to the problem. To see
that the optimal solution is of the form of \eqref{eq:solutionD2}, we will show by
contradiction that a smaller minimum cannot exist. First notice that the
defining property of \eqref{eq:solutionD2} states that $q^*_m$ is either $0$ or
achieves its maximum (infinity norm) at all points except for at most a single one.
Let's then assume that there is a different optimal solution $q^\dag$ with
smaller infinity norm such that it has at least a second non-zero component
smaller than $\|q^\dag\|_\infty$, i.e. there are indexes $a, b \in \{
0,1,\ldots, N \}$ such that $0 < q^\dag_a < \|q^\dag\|_\infty$ and $0 < q^\dag_b < \|q^\dag\|_\infty$.
Now we can define the sets $U^\dag = \big \{ m : q^\dag_m = \|q^\dag\|_\infty \big \}$ and
$L^\dag = \big \{ m : q^\dag_m = 0 \big \}$. Since $q^\dag$ satisfy
\eqref{eq:constrain1} and \eqref{eq:constrain2}:
\begin{align}
  \sum_{m=0}^N g_mq_m^\dag &= \|q^\dag\|_\infty \sum_{m \in U^\dag} g_m  +
  \sum_{m \in (U^\dag \cup L^\dag)^c} g_mq_m^\dag \nonumber \\
  &= (\|q^\dag\|_\infty - \epsilon)\sum_{m \in U^\dag} g_m + \sum_{m \in
  (U^\dag \cup L^\dag)^c} g_mq_m^\dag + \epsilon \sum_{m \in U^\dag} g_m
  \nonumber \\
  &=1 ,
\end{align}
\begin{align}
  \sum_{m=0}^N mg_mq_m^\dag &= \|q^\dag\|_\infty \sum_{m \in U^\dag} mg_m  +
  \sum_{m \in (U^\dag \cup L^\dag)^c} mg_mq_m^\dag \nonumber \\
  &= (\|q^\dag\|_\infty - \epsilon)\sum_{m \in U^\dag} mg_m + \sum_{m \in
  (U^\dag \cup L^\dag)^c} mg_mq_m^\dag + \epsilon \sum_{m \in U^\dag} mg_m
  \nonumber \\
  &= Np,
\end{align}
where $0 < \epsilon < \|q^\dag\|_\infty$. Let's see that in such a case there is another $q^{\dag
\dag}$ with $\|q^{\dag \dag}\|_\infty < \|q^\dag\|_\infty$, defined as
\begin{equation}
  q_m^{\dag \dag} = \left \{ \begin{array}{l l}
  \|q^\dag\|_\infty -\epsilon & \text{if } m \in U^\dag \\
  q^{\dag \dag}_a & \text{if } m = a \\
  q^{\dag \dag}_b & \text{if } m = b \\
  0 & \text{if } m \not \in U^\dag \cup \{ a,b \}  \\
  \end{array} \right.
\end{equation}
where $q_a^{\dag \dag}$, $q_b^{\dag \dag}$ and $\epsilon $ are constants to
be determined satisfying the constrains. Notice that if $q^{\dag
\dag}_a, q^{\dag \dag}_b < \|q^\dag\|_\infty - \epsilon$, then $q^{\dag \dag}$
is a better minimum than $q^\dag$. Additionally, $q^{\dag \dag}$ will satisfy the constrains
of \eqref{eq:constrain1} and \eqref{eq:constrain2} if there are constants $q_a^{\dag
\dag},q_b^{\dag \dag},\epsilon > 0$ such that
\begin{equation}
  \left[ \begin{array}{cc}
    g_a & g_b \\
    ag_a & bg_b \\
  \end{array} \right]
  \left[ \begin{array}{c}
    q_a^{\dag \dag} \\
    q_b^{\dag \dag} \\
  \end{array} \right]
  =
  \left[ \begin{array}{c}
    g_aq_a^{\dag} + g_bq_b^{\dag} + \epsilon \sum_{m \in U^\dag} g_m \\
    ag_aq_a^{\dag} + bg_bq_b^{\dag} + \epsilon \sum_{m \in U^\dag} mg_m \\
  \end{array} \right].
\end{equation}
Since $g_a, g_b \neq 0$ and $a \neq b$, the first matrix is invertible and then
\begin{align}
  \left[ \begin{array}{c}
    q_a^{\dag \dag} \\
    q_b^{\dag \dag} \\
  \end{array} \right]
  &=
  \left[ \begin{array}{c}
    q_a^{\dag} \\
    q_b^{\dag} \\
  \end{array} \right]
  + \epsilon \frac{(g_ag_b)^{-1}}{b-a}
  \left[ \begin{array}{c}
    bg_b\sum_{m \in U^\dag} - g_b\sum_{m \in U^\dag} mg_m \\
    g_a \sum_{m\in U^\dag}mg_m - ag_a\sum_{m \in U^\dag} g_m
  \end{array} \right] \nonumber \\
  &\doteq
  \left[ \begin{array}{c}
    q_a^{\dag} \\
    q_b^{\dag} \\
  \end{array} \right]
  + \epsilon
  \left[ \begin{array}{c}
    s_a \\
    s_b \\
  \end{array} \right] .
\end{align}
Thus, since $0 < q_i^\dag < \|q^\dag\|_\infty$ ($i=a,b$) and $s_a, s_b$ do not depend
on $q^{\dag \dag}$,  there exists $\epsilon > 0$ with $0 < q_i^{\dag
\dag} < \|q^\dag\|_\infty - \epsilon$ ($i=a,b$). As a consequence,
$q^{\dag \dag}$ verifies \eqref{eq:constrain1}, \eqref{eq:constrain2},
\eqref{eq:constrain3} and $\|q^{\dag \dag}\|_\infty < \|q^{\dag}\|_\infty$.
Therefore, $q^\dag$ cannot be the solution of the optimization problem, and then
the optimal solution is the one given in \eqref{eq:solutionD2}.

On the other hand, the expressions of $\gamma$ y $s$ in \eqref{eq:Zys}
follow from the fact that
\begin{equation}
  \left[ \begin{array}{c c}
    \sum_{m \in U} g_m & g_{m^*} \\
    \sum_{m \in U} mg_m & m^* g_{m^*} \\
  \end{array} \right]
  \left[ \begin{array}{c}
    \gamma \\
    s\gamma \\
  \end{array} \right]
  =
  \left[ \begin{array}{c}
    1 \\
    Np \\
  \end{array} \right],
  \label{eq:gamma}
\end{equation}
where $(U \cup L)^c = \{ m^* \}$. Now, for a fixed $m^*$ let us demonstrate
that either $U=\{ 0,1,\ldots, m^*-1\}$ and $L=\{ m^*+1, m^*+2, \ldots, N \}$, or
$U=\{ m^*+1, m^*+2, \ldots, N \}$ and $L=\{ 0,1,\ldots, m^*-1\}$. From \eqref{eq:gamma} we have that
\begin{equation}
  \gamma =  \frac{m^* - Np}{\sum_{m \in U} (m^*-m)g_m},
\end{equation}
and for a fixed $m^*$ and $s$, it is then clear that the way to minimize $\gamma$ is by choosing
$U$ as a set of consecutive indexes as big as possible. Also, with $U$ fixed the
best strategy is to choose $m^*$ as a consecutive index in order to ensure that $\gamma$
will be positive and as small as possible. 
\end{proof}

\bigbreak

It is worth noting that the cases $s = 0$ and $s = 1$ are equivalent with a
proper redefinition of $U$. Indeed, whenever $s = 1$, we can define $U' = U
\cup \{m^*\}$ and then $(U' \cup L)^c = \emptyset$, which by definition means
that $s' = 0$. Conversely, if $s = 0$ (and therefore by definition $(U \cup
L)^c = \emptyset$), we can take any element $m^*$ in $U$ (in this case $U$ can
never be empty, otherwise $q = 0$) and define $U' = U \setminus \{m^*\}$ which
means that now $(U' \cup L)^c = \{m^*\}$ and $s' = 1$. For this reason, in the
remainder of this text we will either use $s = 1$ or $s = 0$ (each with their
own proper definition of $U$) for the same state depending on the convenience
for the calculation at hand.

\bigbreak

\begin{CorollarySM}[Theorem 1 in the main text]
  Given an integer $N$ and a state $\rho$ which satisfies $[\rho, \HS]=0$,
  there exists a subset of energies  ${\mathcal{E}_N^{\rho} \subseteq\mathcal{E}_N}$, 
  a constant $s\in(0,1]$, and at most a single energy $\varepsilon \in \mathcal{E}_N$ such that 
  the state $\optimal$ is defined by the distribution:
  \begin{equation}
    \label{eq:sm:solutionEnergy}
    \lambda_{E} = \left \{ \begin{array}{l l}
     \; \;\;\frac{e^{-\beta E}}{\gamma} & \text{if } E \in \mathcal{E}_N^{\rho} \\
      {s} \;\frac{e^{-\beta \varepsilon}}{\gamma}\,  & \text{if } E = \varepsilon \\
      0 & \text{otherwise }   \\
    \end{array} \right.
  \end{equation}
  with $\gamma$ a normalization constant. The work of formation and the extractable work of $\optimal$ are given by
\begin{align}
  \Wmin(\rho, N) &= \kBT \log\left[ \frac{\ZSN}{\gamma}\right], \\
   \Wext(\rho,N)   &= \kBT \log\left[ \frac{\ZSN}{Z}\right],   \label{eq:sm:Wopt}
\end{align}
where  $Z$ is the partition function of a 
system in a thermal state at temperature $T$ with spectrum given by the set ${\mathcal{E}_N^{\rho}\cup \{\varepsilon}\}$,
and $\gamma=Z- (1-s) \, g_N(\varepsilon) e^{-\beta \varepsilon}$.  

\end{CorollarySM}

\begin{proof}

  It follows from Theorem 1 taking $q_E = \lambda_E e^{\beta E}\ZSN$, defining $\mathcal E_N^\rho = \{ mE_0 : \text{with } m \in U \}$ and $\varepsilon = m^* E_0$ with $m^*$ the only index in $(U \cup L)^c$ (if $(U \cap L)^c = \emptyset$, we redefine $U$ as mentioned after the proof of Theorem 1 to have $(U \cup L)^c = \{m^*\}$ and $s = 1$). The \cwork~is then given by
	\begin{align}
	\Wmin (\rho, N) 
	&= 
  \kBT \log \gamma^{-1}  \nonumber \\
	&= 
   - \kBT \log \left[ \sum_{U} g_m  + s g_{m^*} \right] \nonumber \\
	&= 
        - \kBT \log \left[  \frac{1}{\ZSN} \left ( \sum_{U \cup \{ m^* \} } \binom{N}{m} e^{- \beta m E_0} - (1-s)g_N(\varepsilon) e^{-\beta \varepsilon} \right ) \right] \nonumber \\
      &= \kBT \log\left[ \frac{\ZSN}{\gamma}\right]
 	\end{align}
 	where $Z$ is the partition function of the system restricted to the set of energies $\mathcal E_N^\rho \cup \{ \varepsilon \}$, that is,
 	\begin{equation}
      Z = \ZSN \sum_{ m \in \mathcal E_N^\rho \cup \{ \varepsilon \} } g_m =  
 	\sum_{ m \in \mathcal E_N^\rho \cup \{ \varepsilon \} } \binom{N}{m} e^{-\beta m E_0}.
 	\end{equation}
\end{proof}

\bigbreak

\figref{fig:optsol2d}-a) illustrates the analytical solution for $N=3$.
There we show the dependence of $\gamma$ and $s \gamma$ (\eqref{eq:Zys}) with the local state $\rho = (1-p)\dyad{0} + p\dyad{1}$. 
It can been seen that $\gamma(p)$ is a piecewise linear function and
that the break points coincide abruptly changes of $s$ from $s \approx 0$
to $1$. A similar discontinuity in the derivative of $\gamma$ appears in the derivative of the 
\cwork~(\figref{fig:optsol2d}-b)).

\begin{figure*}[tb]
  \centering
  \includegraphics{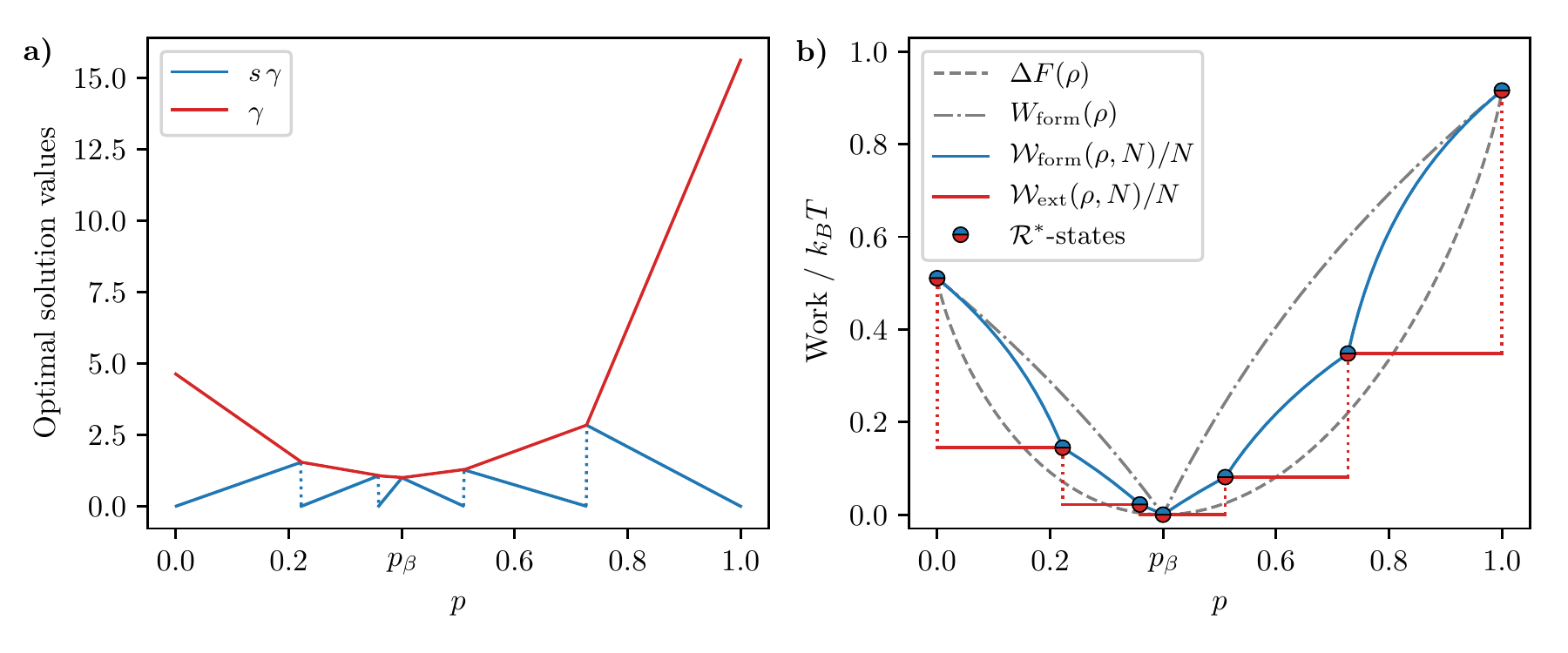}
  \caption{\textbf{Optimal solution for the qubit case and $N=3$ copies}. Panel
    (a) shows the optimal values of $\gamma$ (red curve) and $s\gamma$ (blue
    curve) as a function of the local qubit excited state probability, $p$. The
    dotted lines indicate discontinuities that appear in $s\gamma$. Panel (b)
    shows the work of formation of a single copy $\Wform(\rho)$ (gray
    dash-dotted line), the \cwork~per copy $\Wmin(\rho, N)/N$ (blue line), the
    extractable work from the correlated system per copy $\Wext(\rho, N)/N$
    (red line) and the standard thermodynamic free energy of the local system
    $\Delta F$ (dahsed gray line). The blue and red circles further indicate
    the location of the reversible states. \label{fig:optsol2d}}
\end{figure*}


\section{Bound on the amount of total correlations}

We have shown that correlated states minimize the \cwork. Thus, the creation of correlations can be used to reduce
this energetic cost with respect to the uncorrelated case. The question we address here is whether there is a bound 
on the amount of  total correlations that can be developed while reducing this work cost.
It this easy to show that  total correlations
\begin{equation}
\mathcal I(\rho^{(N)}) > N\beta\,\big(\Wform(\rho)-\Delta F(\rho)\big)
\end{equation}
are necessarily costly, since they cannot be used to reduced the work of formation.
In fact, if a multipartite state $\rho^{(N)}$ is created with total correlations $\mathcal I(\rho^{(N)})/N > (\Wform(\rho)-\Delta F(\rho)) \beta $, 
then $\mathcal I(\rho^{(N)})\, /\beta+N\Delta F(\rho) > N\, \Wform(\rho)$. 
Recalling that $\Delta F(\rho^{(N)})= \mathcal I(\rho^{(N)})\, /\beta+N\Delta F(\rho)$ and
$\Wform(\rho^{(N)})\geq\Delta F(\rho^{(N)})$ then $ \Wform(\rho^{(N)})>N\,\Wform(\rho)$. 


\section{Reversibility and $\criticos$}

As mentioned in the main text, in general the transformations in
the resource theory framework are not reversible, in the sense
that the work of formation of a state is in general strictly larger than the
extractable work \cite{Horodecki2013}. This means that in general we will not be
able to recover all the work invested in the creation of a state from thermal
equilibrium. However, there is a particular set of states where reversible
interconversion can be achieved.

\begin{DefinitionSM}[Reversible states]
  A state $\sigma$ is reversible if its work of formation and extractable work
  coincide; i.e.
  \begin{equation}
    \WE(\sigma) = \Wform(\sigma).
  \end{equation}
  We define as $\mathcal{R}$ the set of all reversible states.
\end{DefinitionSM}

While generic states are irreversible, there is a particularly simple
family of reversible states  which are thermal on a reduced support of the system.

\begin{PropositionSM}[non-trivial reversibles states]
  Let $\system$ be a system with Hamiltonian $\HS$ and energy spectrum
  $\ES$. Let $\tauS$ be the Gibbs state of the system, which is
  characterized by the eigenvalues
  \begin{equation}
    \tauS (E, g) = \frac{1}{\ZS} e^{-\beta E} \quad \forall g \leq \gS(E),
  \end{equation}
  where $\ZS$ is the partition function and $\gS(\cdot)$ the degeneracy. 
  Let $\rho$ be a state of $\system$ with $[\rho, \HS] = 0$ and eigenvalues given by
  \begin{equation}
    \lambda(E, g) = \frac{1}{Z} \tauS(E, g) \mathbbm{1}_{\{E \in \mathcal{E} \}} \quad \forall g \leq \gS(E),
  \end{equation}
  where $\mathcal{E} \subseteq \ES$ is a subset of allowed energies, $\mathbbm{1}$ is the indicator function
  and $Z$ is a normalization constant. Then the family of R\'enyi divergences $\{
  D_\alpha (\rho \| \tauS) \}_{\alpha \geq 0}$ is independent of $\alpha$.
  In particular, the state $\rho$ is reversible.
  \label{prop:pre_reversibilidad}
\end{PropositionSM}

\begin{proof}
  It follows directly from the definition of the R\'enyi divergences \cite{Brandao2015}
  and the form of $\rho$. Given $\alpha \geq 0$
  \begin{align}
    D_\alpha(\rho \| \tauS) &= \frac{1}{\alpha-1} \log \sum_{E \in \mathcal \eS} \sum_{g=1}^{\gS(E)} \lambda(E,g)^\alpha \tauS(E,g)^{1-\alpha} \nonumber \\
    &=  \frac{1}{\alpha-1} \log \sum_{E \in \mathcal{E}} \gS(E) Z^{-\alpha} \tauS(E,g) \nonumber \\
    &= \frac{1}{\alpha-1} \log Z^{-\alpha} \sum_{E \in \mathcal{E}} \gS(E) \tauS(E,g) \nonumber \\
    &= \frac{1}{\alpha-1} \log Z^{1-\alpha} \nonumber \\
    &= - \log Z.
  \end{align}
  In particular,
  \begin{equation}
    D_0(\rho \| \tauS) = \inf_{\alpha > 0}D_\alpha(\rho \| \tauS) =
    D_\infty(\rho \| \tauS) = \sup_{\alpha > 0}D_\alpha(\rho \| \tauS) =
    - \log Z
  \end{equation}
  and therefore $\WE(\rho) = \Wform(\rho) $.
\end{proof}


\subsection{$\mathcal{R}^*$-states}

As we will show, there exists a set of $N$-partite reversible states that are indeed solution to the minimization problem.
The associated reduced states $\rho = (1-p)\dyad{0}+p\dyad{1}$, parametrized by the value of $p$, are defined as $\criticos$
and in \figref{fig:optsol2d}-a) and \figref{fig:optsol2d}-b) are associated to the breaking points in the \cwork.

\begin{DefinitionSM}[$\criticos$]
  Given a number of copies $N$ and a reduced state $\rho(p^*) =
  (1-p^*)\dyad{0}+p^*\dyad{1}$, we say that $\rho(p^*)$ is an $\critico$ if
  $\Wmin(\rho(p^*),N)=\Wext(\rho(p^*),N)$. We call $\mathcal{R}^*(N)$ the set of all $\criticos$ for a
  given $N$.
  \label{def:r-critico}
\end{DefinitionSM}

The set $\mathcal{R}^*(N)$ can be easily characterized by the condition $s=1$ (see Corollary 1).

\begin{PropositionSM}
  Given a system of $N$ qubits, the set of $\criticos$ is determined
  by the $2N+1$ states of the form
  \begin{equation}
    \rho( p^*_k) = (1-p^*_k)\dyad{0}+p_k^*\dyad{1}
  \end{equation}
  where 
  \begin{equation}
  p^*_k =  -\frac{1}{NE_0} \frac{\partial}{\partial \beta} \log Z_k= \frac{\mean{E}_k}{N E_0}.    
  \end{equation}
If $p^*_k \leq p_\beta$, ${Z_k\equiv \sum_{m = 0}^{k} {N \choose m} e^{-\beta m E_0}}$ for  $k=0,..,N$; while if
 $p^*_k > p_\beta$, ${Z_k\equiv  \sum_{m = 0}^{k} {N \choose N-m} e^{-\beta (N-m) E_0}}$ for $k=0,..,N-1$. Here
 the Gibbs state is $\rho(p_\beta)$.
  \label{prop:caracterizacionRestrella}
\end{PropositionSM}

From \eqref{eq:Zys} it follows that the values $p^*_k$ for which $\rho(p^*_k)$ is
an $\critico$ satisfies
\begin{equation}
  p^*_k = \frac{1}{N} \frac{\sum_{U_k} mg_m}{\sum_{U_k} g_m}  
  = -\frac{1}{NE_0} \frac{\partial}{\partial \beta} \log Z_k 
\end{equation}
with $U_k=\{ k, k+1, \ldots, N \}$ or $U_k=\{ 0,1,\ldots, k \}$,  with $k = 0, 1,
\ldots, N$, depending on whether $p^*_k\leq p_\beta$ or $p^*_k\geq p_\beta$. 
Notice that in both cases $k=N$ coincides with the Gibbs state.
Therefore the \cwork~for these states is just 
\begin{equation}
  \Wmin (\rho(p^*_k),N) = - \kBT \log \left [ \frac{Z_k}{\ZSN} \right ] .
\end{equation}

From the family of $\criticos$ it is possible to recover the optimal
solution given by Theorem \ref{teo:Caracterizacion2D} for each $p \in [0,1]$. Let us
consider two different $\criticos$ characterized by consecutive values
$p_k^*$, $p_{k+1}^*$:
\begin{equation}
  p_k^* = \frac{1}{N} \frac{\sum_{m=k}^N mg_m}{\sum_{m=k}^N g_m} \quad , \quad
  p_{k+1}^* = \frac{1}{N} \frac{\sum_{m=k+1}^N mg_m}{\sum_{m=k+1}^N g_m}  ,
\end{equation}
and using the same change of variables that was done in \eqref{eq:change_variables} we can define
\begin{equation}
  q^{*,k}_n = \frac{1}{Z_k}\mathbbm{1}_{\{ k \leq n\}}  \quad , \quad
  q^{*,k+1}_n = \frac{1}{Z_{k+1}}\mathbbm{1}_{\{ k+1 \leq n\}},
\end{equation}
where $\mathbbm{1}$ is the indicator function. Given $p \geq p_\beta$, let $k$ be such that $p_k^* < p < p_{k+1}^*$ and $x \in (0,1)$
such that $p=xp_k^* + (1-x)p_{k+1}^*$. For this value of $p$ the optimal solution is
as follows:
\begin{align}
  q_n &= \frac{k-Np}{\sum_{m=k+1}^N (k-m)g_m} \mathbbm{1}_{ \{ k + 1 \leq n \} } +
  \frac{1}{g_k} \frac{\sum_{m=k+1}^N (Np-m)g_m}{\sum_{k=m+1}^N (k-m)g_m}
  \mathbbm{1}_{\{ k=n \} } \nonumber \\
  &= x \frac{1}{\sum_{m=k}^N g_m}\mathbbm{1}_{\{ k \leq n\}} +
  (1-x)\frac{1}{\sum_{m=k+1}^N g_m}\mathbbm{1}_{\{ k+1 \leq n\}} \nonumber \\
  &= xq^{*,k}_n +(1-x)q^{*,k+1}_n.
  \label{eq:opt_convex_comb}
\end{align}
Therefore, using the same convex combination that relates the value $p$ with
its two nearest $\criticos$, $p_k^*$ and $p_{k+1}^*$, the optimal solution
for $p$ can be written as a convex combination of the optimal solutions of
$p_k^*$ and $p_{k+1}^*$.

%


\subsection{Density}

Another useful property that the $\criticos$ satisfy is that they are
dense in the space of states in the following sense. Let us consider the union
of all the states $\mathcal{R}^*(N)$,
\begin{equation}
  \mathcal{R}^* =   \bigcup_{N \in \mathbb{N}} \mathcal{R}^*(N).
\end{equation}
Then, given $\epsilon > 0$, for each $\rho$ there exists $\rho_\epsilon
\in \mathcal{R}^*$ such that $\|\rho - \rho_\epsilon\|_1 < \epsilon	$. To
prove this, we are going to show that the spacing between two consecutive values
$p_{k}^*$ and $p_{k+1}^*$ for which the state $\rho = (1-p)\dyad{0} +
p\dyad{1}$ is an $\critico$ goes to zero as $\mathcal{O}(1/N)$. 
Since the expression for $p^*_k$ includes the factors $g_m$, for  large enough $N$ we can make
use of asymptotic estimations for the binomial coefficients. From the following
lemma we can conclude the main result of this section.

\begin{LemmaSM}[\cite{asymptotic}]
  Let $X$ a random variable with $X \sim Bi(n, p)$, where  $n \in \mathbb{N}$
  and $p \in [0,1]$. The probability that $X \leq m$ is given by
  \begin{equation}
    B_m(n, p) = \sum_{j=0}^m \binom{n}{m} p^j q^{n-j}
  \end{equation}
  with $q=1-p$. Then, for each $0 \leq m \leq np-\phi(n)$ it holds that
  \begin{equation}
    B_m(n, p) \approx \frac{1}{1-r}\binom{n}{m}p^mq^{n-m}
  \end{equation}
  where $r=qm/(p(n+1-m))$ and $\phi(n)=\mathcal{O}(\sqrt{n})$.
  \label{lema:binomial}
\end{LemmaSM}

\begin{PropositionSM}[Density]
  Let $N$ be the number of correlated copies of the system and $\rho =
  (1-p)\dyad{0} + p\dyad{1}$ the  reduced state. Given two consecutive
  values $p^*_{k}$ and $p^*_{k+1}$ associated to two $\criticos$ such
  that $|k-Np_\beta| \geq \phi(n)$, where $\phi(n)=\mathcal{O}(\sqrt{n})$, it
  follows
  \begin{equation}
    N (p_{k+1}^* - p_{k}^*) \xrightarrow{N \to \infty} 1.
    \label{eq:difpcriticos2}
  \end{equation}
  \label{prop:densidad}
\end{PropositionSM}

\begin{proof}
From Lemma \ref{lema:binomial} it follows
\begin{align}
  \sum_{m=k}^N g_m &= \sum_{m=k}^N \frac{1}{\ZSN}\binom{N}{m}e^{-\beta m
  E_0} \nonumber \\
  &= \sum_{m=k}^N \binom{N}{m} \left( \frac{e^{-\beta E_0}}{\ZS} \right)^m
  \left( \frac{1}{\ZS} \right)^{N-m} \nonumber \\
  &= \sum_{m=0}^{N-m_0} \binom{N}{m} \left( \frac{1}{\ZS} \right)^m \left(
  \frac{e^{-\beta E_0}}{\ZS} \right)^{N-m} \nonumber \\
  &\approx \frac{1}{1-r} \binom{N}{k} \left( \frac{e^{-\beta E_0}}{\ZS} \right)^{k} \left(
  \frac{1}{\ZS} \right)^{N-k} \nonumber \\
  &= \frac{1}{1-r} g_{k},
\end{align}
where $r = e^{-\beta E_0}(n-k)/(k+1)$. Then,
\begin{equation}
  1 + \frac{g_{k}}{\sum_{m=k+1}^N g_m} \approx e^{\beta E_0}
  \frac{k+1}{N-k}.
  \label{eq:aproxProp}
\end{equation}
On the other hand, we for $p^*_k>p_\beta$ we can write:
\begin{equation}
  p_{k}^* = \frac{1}{N}\frac{\sum_{m=k}^N m \binom{n}{m}e^{-\beta
  mE_0}}{\sum_{m=k}^N \binom{n}{m}e^{-\beta mE_0}} = - \frac{1}{E_0}
  \frac{\partial}{\partial \beta} \log \sum_{m=k}^N \binom{n}{m}e^{-\beta m
  E_0}.
\end{equation}
Finally,
\begin{align}
  N (p_{k+1}^* - p_{k}^*) &= \frac{1}{E_0} \frac{\partial}{\partial \beta}
  \log \left[ \frac{\sum_{m=k}^N g_m}{\sum_{m=k+1}^N g_m}  \right] \nonumber \\
  &= \frac{1}{E_0} \frac{\partial}{\partial \beta} \log \left[ 1 +
  \frac{g_{k}}{\sum_{m=k+1}^N g_m} \right] \nonumber \\
  & \approx \frac{1}{E_0} \frac{\partial}{\partial \beta} \left[ \beta E_0 +
  \log \frac{k+1}{N-k} \right] = 1,
\end{align}
where we use that the approximation of \eqref{eq:aproxProp} holds for the
derivative. This follows from the fact that both terms in \eqref{eq:aproxProp} are monotones in $\beta$ and $k$. 
\end{proof}


\subsection{Quasi-thermal states}

In Fig.~2(b) of the main text we can observe that
there is a set of states for which the \cwork~of $N$
copies coincides with the work of formation of a single copy. Since these states
are concentrated around the Gibbs state of the system, we will denote them as
\deit{quasi-thermal states}. For a fixed value of $N$, let $p^*_- , p^*_+ \in
[0,1]$ be the values associated to $\criticos$ given by
\begin{align}
  p^*_{-} &= \frac{1}{N}\frac{\sum_{m=0}^{N-1} mg_m}{\sum_{m=0}^{N-1}g_m} =
  \frac{p_\beta - e^{\beta NE_0}/\ZSN}{1-e^{\beta NE_0}/\ZSN}, \\
  p^*_+ &= \frac{1}{N}\frac{\sum_{m=1}^N mg_m}{\sum_{m=1}^N g_m} =
  \frac{p_\beta}{1-g_0} .
\end{align}
Notice that the length of the interval $[p_-^* , p_+^*]$ decreases
exponentially with the number of copies $N$. 
\begin{PropositionSM}[quasi-thermal states]
  For each local state $\rho = (1-p)\dyad{0}+p\dyad{1}$, with $p \in [p^*_{-},
  p^*_+]$, it holds
  \begin{equation}
    \Wmin(\rho, N) = \Wform(\rho) .
  \end{equation}
  \label{prop:cuasitermico}
\end{PropositionSM}

\begin{proof}
The work of formation of the state $\rho_+ = (1-p^*_+)\dyad{0}+p^*_+ \dyad{1}$
is
\begin{equation}
  \Wform(\rho_+) = \kBT \log \left( \max \{ p_+^* e^{\beta E_0}\ZS ,
  (1-p^*_+)\ZS \} \right) 
  = \kBT \log \left( p_+^* e^{\beta E_0}\ZS \right)
  \label{eq:quasi_workFormation}
\end{equation}
where we use that $p_\beta < p^*_+$. On the other hand, the \cwork~of
$N$ copies of $\rho_+$ is
\begin{align}
  \Wmin (\rho_+, N) &= \kBT \log\left(\frac{1}{\sum_{m=1}^N g_m}\right)  \nonumber \\
  &= \kBT \log \left(\frac{1}{1-g_0}\right) \nonumber\\
  &	= \kBT \log \left(\frac{p_+^*}{p_\beta}\right) \nonumber \\
  &= \kBT \log \left(p_+^* (1+e^{\beta E_0})\right) \nonumber \\ 
  &= \kBT \log \left(p_+^* e^{\beta E_0}\ZS\right) ,
\end{align}
which coincides with \eqref{eq:quasi_workFormation}. Without the term $\kBT
\log$, the \cwork~is linear between $p_\beta$ and $p^*_+$.
Therefore, in light of \eqref{eq:opt_convex_comb} the same result is true for
each $p \in [p_\beta, p_+^* ]$. In the same way we can prove the proposition for $p
\in [p_-^*, p_\beta]$.
\end{proof}


\section{Thermodynamic limit}

In this section we focus on recovering the thermodynamic limit from the states
that minimize the \cwork. We will show that for any
reduced state $\rho$, when the number of copies $N$ goes to infinity,
we recover the classical result for the work of formation $\Wmin (\rho , N)/N \approx \Delta F$, i.e. the
work of formation coincides with the standard free energy difference.
\figref{fig:wformconverg}-a) shows the correlated work of
formation per copy $\Wmin / N$ for different states $\rho$. On the other hand,
in \figref{fig:wformconverg}-b) we show the \cwork~for different numbers of copies $N$ 
of the system as a function of $p$. 
We can see how the \cwork~per copy converges asymptotically to the standard free energy difference.

\begin{figure*}[tb]
  \centering
    \includegraphics{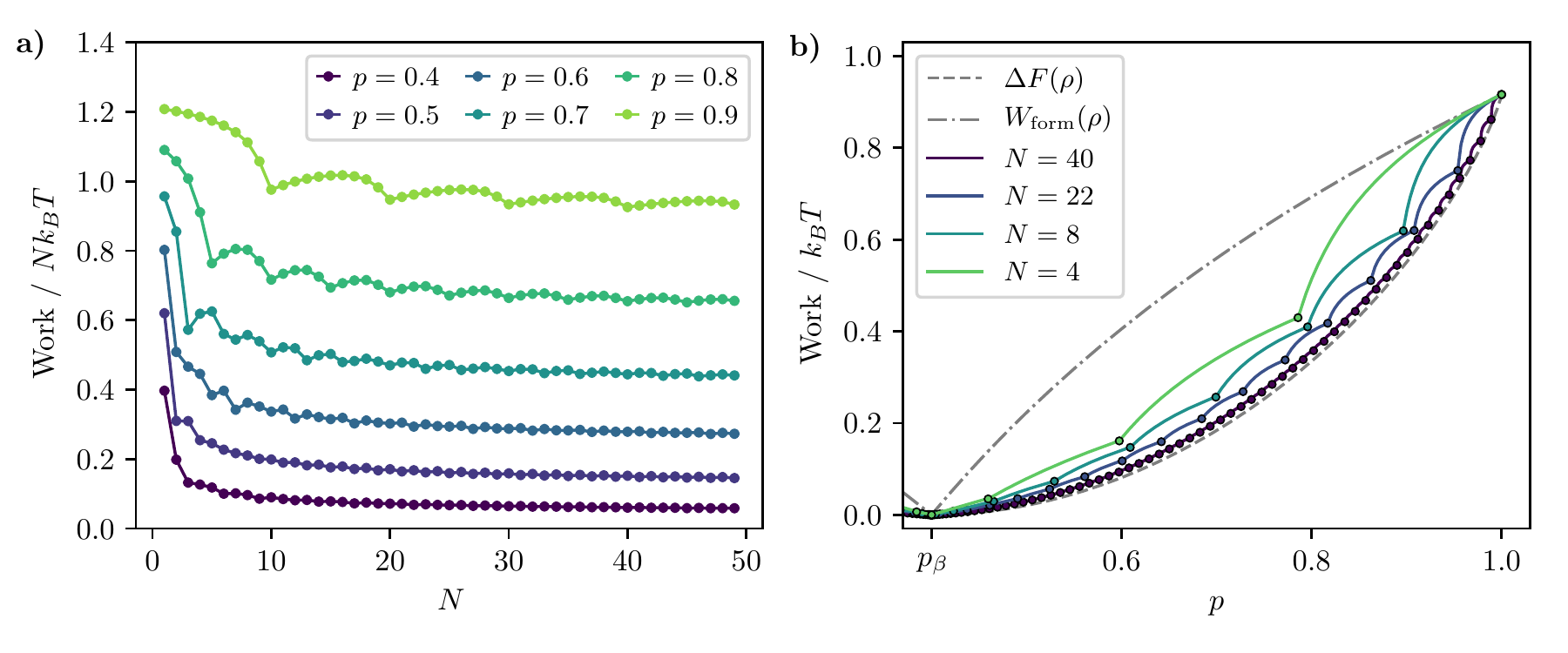}
  \caption{\textbf{Convergence of the \cwork~ as the number of copies increases}. Panel (a) shows the
    \cwork~ $\Wmin(\rho, N)/N$ per copy for different local qubit states
    (parametrized by the excited state probability $p$) as a function of the
    number of copies $N$. We see how the work of formation converges towards
    the standard work of macroscopic thermodynamics $\Delta F(\rho)$ as $N$
    increases. Panel (b) shows the \cwork~ per copy $\Wmin(\rho, N)/N$ for
    different number of copies $N$ as a function of the local qubit state
    (parametrized by the excited state probability $p$). We see how the
    correlated work of formation per copy converges to the standard free energy
    difference. The asymptotic value coincides with the standard work of
    macroscopic thermodynamics $\Delta F(\rho)$ and it is far from the work of
    formation of a single copy. \label{fig:wformconverg}}
\end{figure*}

Let's see how to analytically prove these results in the thermodynamic limit. Since the $\criticos$ 
provide a complete characterization of the \cwork~for each
state and in addition they are dense, it is enough to prove the limit just for the
$\criticos$. Every $\critico$ is of the form $\rho^* = (1-p^*)
\dyad{0} + p^* \dyad{1}$, with
\begin{equation}
  p^* = \frac{1}{N}\frac{\sum_{m \in U} mg_m}{\sum_{m \in U} g_m} .
\end{equation}
Since $U$ is an interval of consecutive energies and the median energy of the
$N$ copies is $Np^*$, we necessarily have $\lceil Np^* \rceil \in U$ or $\lfloor
Np^* \rfloor \in U$. For simplicity, let us just consider that $Np^* \in U$ (the
other case can be treated analogously). Then
\begin{align}
  \Wmin (\rho^*,N) &= \kBT \log \left( \frac{1}{\sum_{m \in U} g_m} \right) \nonumber \\
  &\leq \kBT \log \left( \frac{1}{g_{Np^*}} \right) \nonumber \\
  &= \kBT \log \left(
  \frac{\ZSN e^{\beta Np^* E_0}}{\binom{N}{Np^*}} \right) \nonumber \\
  &= Np^* E_0 + N\kBT \log \ZS - \kBT \log \binom{N}{Np^*}.
\end{align}
Using Stirling's approximation for the binomial coefficients we
have
\begin{equation}
  \log \binom{N}{Np^*} = N S(p^*) + \mathcal{O}(\log N),
\end{equation}
where $S(\cdot)$ is the binary Shannon entropy. On one hand we have
\begin{align}
  \Wmin (\rho^*, N) &\leq Np^* E_0 - \kBT N S(p^*) + \kBT \log \ZS +
  \mathcal{O}(\log N) \nonumber \\
  &= N \left[ F(\rho^*) - F(\tau) \right] +  \mathcal{O}(\log N) \nonumber \\
  &= N \DF(\rho^*) +  \mathcal{O}(\log N) .
\end{align}
On the other hand, it is easy to see that $N\DF(\rho^*) \leq \Wmin (\rho^*, N)$.
Since $\rho^{(N)}_\text{min}$ is reversible and due to the subadditivity of entropy we have
\begin{align}
    \Wmin (\rho^* , N) &= \DF(\rho^{(N)}_\text{min}) \nonumber \\
  &= \tr [\HS^{\otimes N}\rho^{(N)}_\text{min}] - \kBT S(\rho^{(N)}_\text{min}) + \kBT \log
  \ZSN \nonumber \\
  &\geq N \tr [\HS \rho^*] - N\kBT S(\rho^*) + N\kBT \log \ZS \nonumber \\
  &= N \DF(\rho^*) .
  \label{eq:capitulo3cotaW1Winf}
\end{align}
Finally we have that the \cwork~per copy converges to the
 standard free energy difference
\begin{equation}
  \DF(\rho^*) \leq \frac{\Wmin (\rho^*, N)}{N} \leq \DF(\rho^*) +
  \mathcal{O}\left( \frac{\log N}{N} \right) .
\end{equation}
To fully recover the expected thermodynamic limit would furthermore require that
the correlations vanish as $N \to \infty$. We can compute the generalized mutual
information between the correlated state $ \rho^{(N)}_\text{min}$ and the product
state $\rho^{\otimes N}$,
\begin{equation}
  \mathcal{I}\left( \rho^{(N)}_\text{min} \right) = D_1 \left( \rho^{(N)}_\text{min} \|
  \rho^{*\otimes N} \right) = N S(\rho^*) - S(\rho^{(N)}_\text{min}).
\end{equation}
From the reversibility of the $\criticos$ it is possible to rewrite the
mutual information as a function of the \cwork~$\Wmin$ and
the standard free energy difference $\DF$,
\begin{align}
  \mathcal{I}\left(\rho^{(N)}_\text{min}\right) &= \beta [ F(\rho^{(N)}) - NF(\rho^*) ] \nonumber \\
  &= \beta [\DF(\rho^{(N)}) - N\DF(\rho^*)] \nonumber \\
  &= \beta [\Wmin(\rho^*, N) - N\DF(\rho^*)] = \mathcal{O}\left( \log N \right) .
\end{align}
We can then conclude that the amount of correlations scales as $\log N$ and
therefore the correlations per copy vanish, thus the subsystems are weakly
correlated in the macroscopic limit.



\section{Generalization to systems of arbitrary dimension}

In this section we will generalize the above results to systems
of arbitrary finite dimension $D$. As we mentioned before, the analytical
characterization of the minimization problem is quite difficult and requires 
the analytic expression of the degeneracy $g_N(E)$ for each energy
$E$ of the $N$-partite system, which cannot be known for arbitrary Hamiltonians.
However, it can be proven that the optimal solution has the same form as the one given
in Theorem \ref{teo:Caracterizacion2D}. Furthermore, the existence of
quasi-thermal states, their reversibility, density and the thermodynamic limit
still hold for all $D$.

Let us consider the optimization problem presented in
\eqref{eq:minimization.posta}, \eqref{eq:constrain1.posta} and
\eqref{eq:constrain2.posta} for systems of dimension $D$
\begin{align}
  \min_{q} \quad \|q\|_\infty \nonumber \\
  \text{s.t.} \qquad \sum_{E \in \mathcal{E}_N} \frac{g_{N-1}(E-E_d)}{\ZSN}
  e^{-\beta E} q_E &= p_d \quad &\forall d = 1, 2, \ldots, D \\
  q_E &\geq 0 \quad &\forall E \in \mathcal{E}_N. \nonumber
\end{align}
Notice that the form of the optimal solution introduced in Theorem
\ref{teo:Caracterizacion2D} can be written as
\begin{equation}
	q_i = \gamma \mathbbm{1}_{\{ i\in U \}} + s \gamma \mathbbm{1}_{\{i = m^* \}} .
	\label{eq:optimal_2D_indicator}
\end{equation}
Let $A \in \R^{D \times M}$ (where $M = \#(\mathcal{E}_N)$ is the cardinality of
the set of energies of the $N$-partite system) be the matrix where each row
corresponds to an energy $E_d$ of a single system and each column to an energy $E$ of
the $N$-partite system, and the matrix elements are equal to $g_{N-1}(E-E_d) e^{-
\beta E} / \ZSN$. The following theorem proves the natural generalization of
\eqref{eq:optimal_2D_indicator}, where now the optimal solution either realizes
its infinity norm or it is equal to zero except in at most $D-1$ points.

\begin{TheoremSM}
  Consider the optimization problem given by
  \begin{align}
    \min_{q \in \R^M} \quad \|q\|_\infty \nonumber \\
    \text{s.t.} \quad  Aq &= p  \label{eq:minimizar_general_teorema} \\
    q &\geq 0 \nonumber
  \end{align}
  where $p \in \R_{\geq 0}^D$ and $A \in \R_{\geq 0}^{D \times M}$, such that there is at least one feasible solution to the constrains. Then, there
  exists a solution, not necessarily unique, of the form
  \begin{equation}
 	  q_i = \gamma \mathbbm{1}_{ \{ i \in U \} } + \sum_{j=1}^{D-1} \gamma s_{j}
 	  \mathbbm{1}_{ \{i = n_j \} }
    \label{eq:solucion_teorema_general}
 	\end{equation}
 	where $s_{1}, s_{2}, \ldots , s_{D-1} \in [0,1]$; $\gamma \in \R_{> 0}$; $U$ is a set of indexes
 	where $q$ realizes its infinity norm; and  $n_1, n_2, \ldots, n_{D-1} \in \{
 	1,2,3,\ldots , M \}$.
  \label{teo:generalizacion}
\end{TheoremSM}

\begin{proof}
Let $\mathcal{A}$ be the set of optimal solutions of
\eqref{eq:minimizar_general_teorema}. Since $\|\cdot\|_\infty$ is a distance and
the constrains $Aq = p$ and $q \geq 0$ define a convex, bounded and closed set in
$\R^M$, it follows that there is at least one solution and then $\mathcal{A}$ is
non-empty. For each $q$ we define
\begin{equation}
  M(q) = \left\{i : q_i \neq 0 \,\land\, q_i \neq \|q\|_\infty \right\},
\end{equation}
and for each $p$ let us consider the quantity
\begin{equation}
  m(p) = \min \biggl\{ |M(q)| : q
  \text{ is a solution of \eqref{eq:minimizar_general_teorema}} \biggr\}.
\end{equation}
Then, the main statement of the Theorem is equivalent to proving that $m(p) \leq
D-1$ for every $p$. Let $q^\dag \in \mathcal{A}$ be such that $|M(q^\dag)|=m(p)$
and assume that $m(p) \geq D$. Additionally, let $\{ i_1, i_2, \ldots , i_D \}
\subset M(q^\dag)$ and consider the matrix
\begin{equation}
  A_M =	\left [
    \begin{array}{c|c|c|c}
      & & &  \\
      a_{i_1} & a_{i_2} & \cdots & a_{i_D} \\
      & & & \\
    \end{array} \right ],
\end{equation}
where $a_{i_1}, \ldots, a_{i_D}$ are the columns $i_1, \cdots, i_D$ of the
matrix $A$. We will now show that $A_M$ cannot be invertible. If $A_M$ were invertible,
then it would be possible to explicitly construct a solution $q^{\dag \dag}$
with $\|q^{\dag \dag}\|_\infty < \|q^\dag\|_\infty$ and given by
\begin{equation}
  q_m^{\dag \dag} = \left\{ \begin{array}{l l}
    \|q^*\|_\infty -\epsilon &
    \text{if } q^\dag_m = \|q^\dag\|_\infty \\
    x_1 & \text{if } m = i_1 \\
    \vdots &\\
    x_D & \text{if } m = i_D \\
    q^\dag_m & \text{otherwise, } \\
  \end{array} \right.
\end{equation}
where $\epsilon > 0$ and $x_1, \ldots , x_D$ are chosen such that the constrains
are satisfied and $x_1, \ldots, x_D < \|q^\dag\|_\infty - \epsilon$. $q^{\dag
\dag}$ is a feasible solution if
\begin{align}
  p = Aq^{\dag \dag} &= A_U \left [ \begin{array}{c}
    \|q^*\|_\infty -\epsilon \\
    \vdots \\
    \|q^*\|_\infty -\epsilon
  \end{array} \right ]   +
  A_M  \left [ \begin{array}{c}
    x_1 \\
    \vdots \\
    x_D
  \end{array} \right ] \nonumber \\
	&= Aq^\dag +
	A_M  \left [ \begin{array}{c}
    x_1 - q_{i_1}^\dag \\
    \vdots \\
    x_D - q_{i_D}^\dag
  \end{array} \right ] -
	\epsilon \sum_{m \in U} a_{m},
\end{align}
where $A_U$ is the restriction of $A$ to the columns $a_{m}$ for which $q^\dag_m =
\|q^\dag\|_\infty$. Since $Aq^\dag = p$ and assuming that $A_M$ is invertible, then 
\begin{equation}
  \left [ \begin{array}{c}
    x_1 - q_{i_1}^\dag \\
    \vdots \\
    x_D - q_{i_D}^\dag
  \end{array} \right ]
  = \epsilon A_M^{-1}  \sum_{m \in U} a_{m} .
	\label{eq:proof_Dgeneral}
\end{equation}
Since the right side of \eqref{eq:proof_Dgeneral} is fixed and $0 < q_{i_j}^\dag <
\|q^\dag\|_\infty$, it follows that there exists $\epsilon > 0$ small enough
such that $0 < x_i < \|q^\dag\|_\infty$. Then $\|q^{\dag \dag}\|_\infty =
\|q^{\dag }\|_\infty - \epsilon < \|q^{\dag}\|_\infty$ and $q^\dag$ cannot be an
optimal solution, which contradicts our hypothesis. Therefore $A_M$ is not
invertible.

Since $A_M$ is not invertible, there exists a non-trivial solution of the
equation $A_M x = 0$. Then, it is possible to define a new solution to $Aq = p$
in the following way
\begin{equation}
	q_m^{\dag \dag} = \left \{ \begin{array}{l l}
	\|q^*\|_\infty  & \text{if } q^\dag_m = \|q^\dag\|_\infty \\
	q^\dag_m + \delta x_1 & \text{if } m = i_1 \\
	\vdots &\\
	q^\dag_m + \delta x_D & \text{if } m = i_D \\
	q^\dag_m & \text{elsewhere. }\\
	\end{array} \right.
\end{equation}
Then it is immediate that there exists $\delta > 0$ and an index $i_k$ such that
$q_{i_k}^{\dag \dag} = 0$ or $q_{i_k}^{\dag \dag} = \|q^\dag\|_\infty$ and $0 <
q^{\dag \dag}_m < \|q^\dag\|_\infty$ for all $m \neq i_k$. Thus, $M(q^{\dag
\dag}) = m(p) - 1$, which contradicts the definition of $m(p)$. Finally we
conclude that $m(p) \leq D-1$ for all $p$.
\end{proof}

\bigbreak

As in the case of $D=2$, Theorem \ref{teo:generalizacion} claims that the
state that minimizes the work of formation is a renormalization of the Gibbs
state over a restricted subset energies, except for at most $D-1$ energy levels that
escape the rule. \figref{fig:WcorrDgeneral} shows the value of the
\cwork~for $D=3$ obtained numerically for each $p = (p_1, p_2, p_3)$, and the
values of $\gamma$, $s_1 \gamma$ and $s_2 \gamma$ for a particular path of states.

\begin{figure*}[tb]
  \centering
  \includegraphics{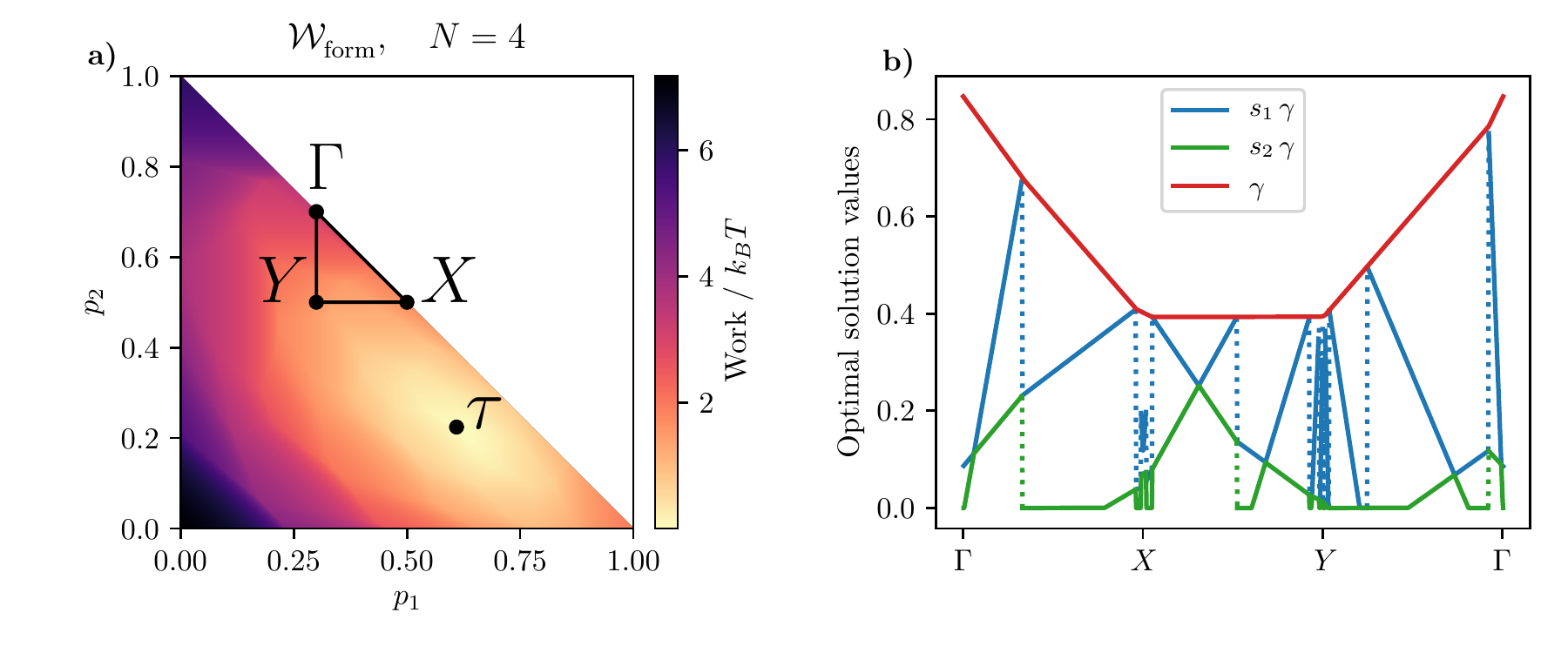}
  \caption{\textbf{Optimal solution for the $D = 3$ case}. In panel (a) we
  show the correlated work of formation $\Wmin$ for $N=4$ copies of a system of
  dimension $D=3$ for every reduced local state of the system. $p_1$, $p_2$ and
  $p_3$ correspond to the occupation of the energy levels $E_1$, $E_2$ and
  $E_3$, respectively. Panel (b) shows the values of $\gamma$, $s_1 \gamma$ and
  $s_2 \gamma$ of the optimal solution given by Theorem
  \ref{teo:generalizacion} for a particular path between states $\Gamma \to X
  \to Y \to \Gamma$. \label{fig:WcorrDgeneral} }
  \end{figure*}

\subsection{$\criticos$}

Here we will show the generalization  of  the $\criticos$ to
arbitrary dimension $D$. From now on, for simplicity, we will index the coordinates of $q$
and the columns of $A$ using the associated energy $E$ of the $N$-partite system.  
We define the set of $\criticos$ as the solutions of the optimization problem
of \eqref{eq:minimizar_general_teorema} which are characterized by a subset of
energies $\mathcal{E} \subset \mathcal{E}_N$ and which take the form
\begin{equation}
	(q_\mathcal{E}^{*})_E = \gamma_\mathcal{E} \mathbbm{1}_{ \{ E \in \mathcal{E} \} },
	\label{eq:r_critico_general}
\end{equation}
where
\begin{equation}
	\gamma_\mathcal{E} = \frac{1}{\norm{\sum_{E \in \mathcal{E}} a_{E}}_1}.
\end{equation}
An equivalent definition of the $\criticos$ is given by the solutions
that appear in Theorem \ref{teo:generalizacion} which also satisfy $s_1 = \ldots =
s_{D-1} = 0$. Each $\critico$ of the form given by \eqref{eq:r_critico_general} has
an associated vector $p^*_\mathcal{E}$ for which $q_\mathcal{E}^*$ is solution
of $Aq^*_\mathcal{E}=p^*_\mathcal{E}$ and $q^*$ is optimum, given by
\begin{equation}
	p^*_\mathcal{E} = \frac{1}{\norm{\sum_{E \in \mathcal{E}} a_E}_1} \sum_{E \in
	\mathcal{E}} a_E.
\end{equation}

Notice that Proposition \ref{prop:pre_reversibilidad} is valid even for systems
of arbitrary dimension $D$. Then it follows that the states
defined in \eqref{eq:r_critico_general} are also reversible. On the other hand, for
$D=2$, the unique subset of energies $\mathcal{E}$ for which
we find $\criticos$, are $\{ E : E \geq kE_0 , k = 0,1,\ldots , N  \}$
and $\{ E : E \leq kE_0 , k = 0,1,\ldots , N \}$. However, for $D>2$ and
arbitrary energies $E_1, E_2, \ldots, E_D$ it is not possible to characterize
the sets $\mathcal{E}$ and therefore we cannot find an analytical expression for
the $\criticos$. However, we can give an implicit construction of them.

\begin{PropositionSM}
	Consider $N$ copies of a system of dimension $D$ and let $p \in \R^D$ be the
	occupation of the reduced state $\rho$ of the system. Then, there exists a state
	of the $N$ copies characterized by a vector $q(p)$ that minimizes the work of
	formation such that there are $p_1^*, p^*_2, \ldots, p^*_D \in \R^D$
	probability vectors associated with $\criticos$ $q(p^*_1), q(p^*_2),
	\ldots, q(p^*_D)$ and
  \begin{align}
    p &= \lambda_1 p_1^* + \lambda_2 p_2^* + \ldots + \lambda_D p^*_D, \\
    q(p) &= \lambda_1 q(p^*_1) + \lambda_2 q(p^*_2) + \ldots + \lambda_D
    q(p^*_D),
    \label{eq:suma_convexa_p}
	\end{align}
	where $\lambda_i\in\R_{\ge0}$ with
	$\lambda_1 + \lambda_2 + \ldots + \lambda_D = 1$.
\end{PropositionSM}

\begin{proof}
Let $s_1, s_2, \ldots, s_{D-1}$ be ordered from highest to lowest and with
associated energies $\varepsilon_1, \varepsilon_2, \ldots, \varepsilon_{D-1}$,
respectively. Now we denote by $\mathcal{E}_j$ the set of the first $j$ energies
$\varepsilon_1, \varepsilon_2, \ldots, \varepsilon_j$. Then:
\begin{align}
	q_E &= \gamma_\mathcal{E} \mathbbm{1}_{\{ E \in \mathcal{E} \}} + \sum_{j=1}^{D-1}
	\gamma s_j \mathbbm{1}_{\{ E = \varepsilon_j \}} \nonumber \\
	&= \gamma_\mathcal{E}(1-s_1) \mathbbm{1}_{\{ E \in \mathcal{E} \}} \nonumber \\
	&+ \gamma_\mathcal{E}(s_1-s_2) \mathbbm{1}_{\{ E \in \mathcal{E} \} \cup \{ E \in \mathcal{E}_1 \} }  \nonumber \\
	& \qquad \vdots \nonumber \\
	&+ \gamma_\mathcal{E}(s_{D-2}- s_{D-1}) \mathbbm{1}_{\{ E \in \mathcal{E} \} \cup \{ E \in
	\mathcal{E}_{D-2} \}  } \nonumber \\
	&+ \gamma_\mathcal{E} s_{D-1} \mathbbm{1}_{\{ E \in \mathcal{E} \} \cup \{ E \in
	\mathcal{E}_{D-1} \}  } \nonumber \\
	&= \lambda_1 (q^*_\mathcal{E})_E + \lambda_2 (q^*_{ \mathcal{E} \cup
	\mathcal{E}_1 } )_E + \ldots + \lambda_D (q^*_{ \mathcal{E} \cup
	\mathcal{E}_{D-1} })_E,
\end{align}
where $\lambda_1, \lambda_2, \ldots, \lambda_D$ are all positive values that
satisfy $\lambda_1 + \lambda_2 + \ldots + \lambda_D = 1$ and are given by
\begin{gather}
	\lambda_1 = \gamma_\mathcal{E}(1- s_1) \norm{\sum_{E \in \mathcal{E}} a_E}_1 \quad ,
	\quad \lambda_D = \gamma_\mathcal{E} s_{D-1} \norm{\sum_{E \in \mathcal{E} \cup \mathcal{E}_{D-1} }
	a_E}_1 \quad \text{and} \\
  \lambda_j = \gamma_\mathcal{E} \left( s_{j-1}- s_{j} \right) \norm{\sum_{E \in \mathcal{E} \cup
  \mathcal{E}_{j-1} } a_E}_1 \quad \text{for } j = 2,3, \ldots, D-1.
\end{gather}
Then, each solution in Theorem \ref{teo:generalizacion} can be written
as a linear convex sum of states of the form \eqref{eq:r_critico_general}
\begin{equation}
	q =  \lambda_1 q^*_\mathcal{E} + \lambda_2 q^*_{\mathcal{E}\cup \mathcal{E}_1 } +
	\ldots + \lambda_D q^*_{\mathcal{E} \cup \mathcal{E}_{D-1}} .
	\label{eq:suma_convexa_r_criticos}
\end{equation}
What remains is to prove that each term of the convex sum
in \eqref{eq:suma_convexa_r_criticos} is associated with an $\critico$. Let's
assume that $q^*_\mathcal{E}$ is not an $\critico$. Then there exists $\tilde{q}$ with
$\|\tilde{q}\|_\infty < \|q^*_\mathcal{E}\|_\infty$ and $A\tilde{q} =
Aq^*_\mathcal{E}$. It is then clear that
\begin{equation}
	q^\dag = \lambda_1 \tilde{q} + \lambda_2 q^*_{\mathcal{E}\cup \mathcal{E}_1} +
	\ldots + \lambda_D q^*_{\mathcal{E} \cup \mathcal{E}_{D-1} }
\end{equation}
satisfies $Aq^\dag = A q$. On the other hand,
\begin{align}
	\|q^\dag\|_\infty &\leq \lambda_1 \| \tilde{q} \|_\infty + \ldots +
	\lambda_D \|q^*_{\mathcal{E} \cup \mathcal{E}_{D-1}}\|_\infty \nonumber \\
  &< \lambda_1 \| q^*_\mathcal{E} \|_\infty + \ldots + \lambda_D \|q^*_{\mathcal{E} \cup
  \mathcal{E}_{D-1} }\|_\infty \nonumber \\
	&= \lambda_1 \frac{1}{\norm{\sum_{E \in \mathcal{E}} a_E}_1} + \ldots +
	\lambda_D \frac{1}{\norm{\sum_{E \in \mathcal{E} \cup \mathcal{E}_{D-1} } a_E
	}_1} \nonumber \\
  &= \gamma_\mathcal{E} = \|q\|_\infty ,
\end{align}
that is $\|q^\dag\| < \|q\|_\infty$, which contradicts the definition of $q$. In the same way we deduce that all
$q^*_{\mathcal{E} \cup \mathcal{E}_1}, \ldots, q^*_{\mathcal{E} \cup
\mathcal{E}_{D-1}}$ are $\criticos$.
\end{proof}

\bigbreak

The existence of quasi-thermal states can be obtained numerically. Just like in the
$D=2$ case, there is a region surrounding the Gibbs state of the system for
which the \cwork~of $N$ copies coincides with the work of
formation of a single copy (see for example Supplementary Figure \ref{fig:WcorrDgeneral}).


\subsection{Density of reversible states}

Without an analytic expression for the $\criticos$
and their respective occupation levels $p^*$, we cannot prove rigorously
that the spacing between the vales of $p^*$ goes to zero when $N \rightarrow
\infty$. Notice that even for the simplest case $D=2$ it was necessary to use an
asymptotic approximation of the binomial coefficients. However, we now present
an argument for assuming that the distance between the terms that appear in
\eqref{eq:suma_convexa_p} (which generalize $p^*_{k}$ and $p_{k+1}^*$ of
Proposition \ref{prop:densidad}) should converge to zero when $N$ goes to
infinity. In the worst case scenario, the distance in $\| \cdot \|_1$ between
close values of $p^*$ is
\begin{align}
  \norm{ p^*_{\mathcal{E} \cup \mathcal{E}_{D-1} } - p^*_{\mathcal{E}}}_1  &=
	\norm{ \frac{1}{\norm{\sum_{E \in \mathcal{E} \cup \mathcal{E}_{D-1} } a_E}_1}
	\sum_{E \in \mathcal{E} \cup \mathcal{E}_{D-1}} a_E - \frac{1}{\norm{\sum_{E
	\in \mathcal{E}} a_E}_1} \sum_{E \in \mathcal{E}} a_E
	 }_1 \nonumber \\
	&= \norm{ \frac{\sum_{E \in \mathcal{E}_{D-1}} a_E}{ \norm{\sum_{E \in
	\mathcal{E} \cup \mathcal{E}_{D-1} } a_E}_1 } - \frac{\norm{\sum_{E \in
	\mathcal{E}_{D-1}} a_E }_1 \sum_{E \in \mathcal{E}} a_E }{\norm{\sum_{E \in
	\mathcal{E} \cup \mathcal{E}_{D-1} } a_E}_1 \norm{\sum_{E \in \mathcal{E}}
	a_E}_1 } }_1 \nonumber \\
	& \leq 2 \frac{\norm{\sum_{E \in \mathcal{E}_{D-1}} a_E }_1}{\norm{\sum_{E \in
	\mathcal{E} \cup \mathcal{E}_{D-1} } a_E}_1},
\end{align}
where we use that all $a_E$ have non-negative entries.
Notice that the amount of terms in the numerator is at most $D-1$, while the
number of terms in the denominator is as large as the size of the set
$\mathcal{E}$ of energies that defines the $\critico$. It is reasonable to
assume that the size of $\mathcal{E}$ increases with the number of copies and then
\begin{equation}
 \lim_{N \to \infty}	\norm{ p^*_{\mathcal{E} \cup \mathcal{E}_{D-1} } - p^*_{\mathcal{E}}}_1 = 0.
\end{equation}

Instead of proving the above statement, that would be the direct analog of the
proof provided for the 2-dimensional case for arbitrary dimension, we will
instead prove that general reversible states are dense, in the sense that
given a probability vector $p \in \mathbb{R}^D$, then there
is a sequence of reversible states $ \{ \rho^{(N)}_{\rm rev} \}_N$ such that
\begin{equation}
  \lim_{N \to \infty} \max_{j = 1, 2, \ldots, N} \left \| \tr_{-j}\left[\rho^{(N)}_{\rm rev}\right] - \rho(p) \right \|_1 = 0,
\end{equation}
with  $\rho(p) = \sum_{d=1}^D p_d \dyad{E_d}$. Unlike before, we will not require that these states also be a solution to the
work of formation optimization problem. Nonetheless, this weaker result will still
be enough to prove the thermodynamic limit.

Given the probability vector $p \in \mathbb R^D$, for each $d = 1, 2, \ldots, D$ we define
\begin{align}
  n_i^d &= \floor{ p_i N}, \qquad\qquad\qquad\qquad\qquad \text{for } i \in \{1,\dots,D\} \setminus \{d\}, \\
  n_d^d &= N - \sum_{i \neq d} n_i^d = N - \sum_{i \neq d} \floor{p_i N},
\end{align}
and the following states $\sigma_d$
\begin{equation}
  \sigma_d = \frac{1}{C_d} \sum_{\text{permutations}} \underbrace{\dyad{E_1} \otimes
  \ldots \otimes \dyad{E_1}}_{n_1^d \text{ times}} \otimes \ldots \otimes
  \underbrace{\dyad{E_D} \otimes \ldots \otimes \dyad{E_D} }_{n_D^d  \text{ times}},
  \label{eq:sigmad}
\end{equation}
where the sum runs over all the possible permutations of the $N$ copies, and $C_d$
is a normalization constant given by
\begin{equation}
	C_d = \frac{N !}{(n_1^d)! (n_2^d)! \ldots (n_D^d)!}.
\end{equation}
Notice that the states $\sigma_d$ are eigenstates of the total Hamiltonian with energy $E^{(d)} = n_1^d E_1 + \ldots + n_D^d E_D$,
and their reduced state is
\begin{equation}
	\tr_{-j}(\sigma_d) = \sum_{i=1}^D \frac{n_i^d}{N} \dyad{E_i} \quad \forall j = 1,
	2, \ldots , N.
\end{equation}
which is clear that converges to $\rho(p)$ as $N \to \infty$. Additionally, it is easy to see that all $\sigma_d$ are reversible,
since only a single energy level is equally populated for each $d$. Thus, we have defined a family of reversible states whose reduced state
is equal to $\rho(p)$ (in the case where $Np_i \in  \mathbb{N}$, otherwise it converges to  $\rho(p)$). In the next section we will show
that while for finite $N$ these reversible states are not optimal (in the sense that they do not minimize the \cwork~for finite $N$) their work of formation
per copy converges to the standard free energy difference in the thermodynamic limit.

Before showing how one can recover standard results in the thermodynamic limit, we will first make a connection 
between these states and the optimal ones through the following lemma.

\begin{LemmaSM}
	There is at least one of the energies $E^{(d)}$ that belongs to 
	the support $\mathcal E_N$ of energies of the $\critico$ $\rho^{(N)}_\text{min}$.
	\label{lema:NE_in_support}
\end{LemmaSM}

\begin{proof}
Let's assume that $E^{(d)} \not \in \mathcal E_N$ for all $d=1, 2, \ldots, D$. Consider the state given by
	\begin{equation}
	\sigma = \sum_{d=1}^D \lambda_d \, \sigma_d \quad \text{ with } \quad \lambda_d = \frac{ Np_d - \floor{Np_d} }{ \sum_{i=1,\dots,D} Np_i - \floor{Np_i} }.
	\end{equation}
In the case where $Np_d \in \mathbb N$ for all $d=1, 2, \ldots, D$, we define $\sigma = \sigma_1$. Since $\tr_{-j} ( \sigma ) = \rho(p)$, given $\epsilon > 0$ we can define a new state ${\rho}^{(N)}$
locally equivalent to $\rho(p)$:
\begin{equation}
	{\rho}^{(N)} = (1-\epsilon)\rho^{(N)}_\text{min} + \epsilon \sigma .
\end{equation}
Let's see now that the work of formation of ${\rho}^{(N)}$ is lower than the work of formation of $\rho^{(N)}_\text{min}$.
Since by hypothesis $\sigma$ and $\rho^{(N)}_\text{min}$ have disjoint support, the work of formation of $\rho^{(N)}$ is
\begin{equation}
	\Wform({\rho}^{(N)}) = \kBT \log \max \left\{ (1-\epsilon)
  \frac{\lambda_E \ZSN}{e^{-\beta E}} , \frac{\epsilon \lambda_1 \ZSN}{C_1
  e^{-\beta E^{(1)} }}, \ldots, \frac{\epsilon \lambda_D \ZSN}{C_D
	e^{-\beta E^{(D)}}} \right\},
\end{equation}
where $\lambda_E$ are the eigenvalues of $\rho^{(N)}_\text{min}$. Then,
there exists $\epsilon > 0$ small enough such that the maximum is realized in
the first element and then
\begin{align}
	\Wform({\rho}^{(N)}) &= \kBT \log \left[ (1-\epsilon)
  \frac{\lambda_E \ZSN}{e^{-\beta E}} \right] \nonumber \\
  &= \Wform(\rho^{(N)}_\text{min}) + \kBT \log (1-\epsilon) <
  \Wform(\rho^{(N)}_\text{min}),
\end{align}
which is a contradiction that comes from assuming that $E^{(d)} \not
\in \mathcal{E}_{N}$ for all $d=1, 2, \ldots, D$.
\end{proof}

\bigbreak

In the case where $Np_i \in \mathbb{N}$, this lemma states that the mean energy $Np_1 E_1 + \ldots + Np_D E_D$ is part of the support of energies of the optimal state $\rho^{(N)}_\text{min}$. In the other cases, we can find an energy as close as we want to the mean energy as 
$N \to \infty$.


\subsection{Thermodynamic limit}

We will first show that as $N$ increases the \cwork~per copy 
$\Wmin (\rho, N) / N$ converges to the standard free
energy difference for systems of arbitrary dimension $D$, that is,
\begin{equation}
	\frac{\Wmin (\rho, N)}{N} \xrightarrow{N \to \infty} \Delta F(\rho).
	\label{eq:convergeDgenerico}
\end{equation}
Moreover, the convergence rate is
$\mathcal{O}(\log N / N)$. Just like as in the case $D=2$, let us consider the $\criticos$.
In this case their \cwork~is given by
\begin{align}
	\Wmin (\rho^*, N) &= + \kBT \log \gamma_\mathcal{E} \nonumber \\
	&= -\kBT \log \norm{\sum_{E \in \mathcal{E}} a_E}_1 \nonumber \\
	&= -\kBT \log \left[ \sum_{E \in \mathcal{E}} \sum_{i=1}^D \frac{g_{N-1}(E -
  E_i)}{\ZSN} e^{-\beta E}  \right] \nonumber \\
  &= -\kBT \log \left[ \sum_{E \in \mathcal{E}} \frac{1}{\ZSN} e^{-\beta E}
  \sum_{i=1}^D g_{N-1}(E-E_i) \right] \nonumber  \\
  &= -\kBT \log \left[ \sum_{E \in \mathcal{E}} \frac{g_N(E)}{\ZSN} e^{-\beta
	E}  \right],
  \label{eq:trabajoCorrGeneralCota}
\end{align}
where we use that $g_N(E) = \sum_{i=1}^D g_{N-1}(E-E_d)$. Let's start with the case where 
$Np_1 , Np_2 , \ldots , Np_D \in \mathbb{N}$, then by
Lemma \ref{lema:NE_in_support} we have that $N \langle E \rangle \in \mathcal E_N$, where 
$\langle E \rangle = \sum_{i=1}^D p_i E_i$ is the mean energy of a single copy.
In such a case, we can bound the sum in \eqref{eq:trabajoCorrGeneralCota} by a single term in the
following way
\begin{align}
	\Wmin (\rho^*, N) &= -\kBT \log \left[ \sum_{E \in \mathcal{E}}
  \frac{g_N(E)}{\ZSN} e^{-\beta E}  \right] \nonumber \\
  & \leq -\kBT \log \left[ \frac{g_N(N\langle E \rangle)}{\ZSN} e^{-\beta N
  \langle E \rangle} \right] \nonumber \\
  &= N\langle E \rangle + N \kBT \log \ZS -\kBT \log \left[ g_N(N\langle E
	\rangle) \right].
  \label{eq:trabajoCorrGeneralCota2}
\end{align}
Notice that the degeneracy $g_N(N\langle E \rangle)$ is grater than or equal to the
number of configurations of $N$ copies with $Np_1$ copies with energy $E_1$,
$Np_2$ copies with energy $E_2$, etc. That is, it holds
\begin{align}
	g_N \left( N\langle E \rangle \right) & \geq \binom{N}{Np_1} \binom{N-Np_1}{Np_2} \cdots
	\binom{N-Np_1 - \ldots - Np_{D-1}}{Np_D} \nonumber \\
  &= \frac{N!}{(Np_1)!(Np_2)! \ldots (Np_D)!} .
\end{align}
Applying the logarithm and using Stirling's approximation we get
\begin{equation}
  g_N(N \langle E \rangle) \geq -\sum_{i=1}^D N p_i \log p_i + \mathcal{O}(\log N)= N S(\rho^*) +
	\mathcal{O}(\log N),
\end{equation}
where $S(\rho^*)$ is the von Neumann entropy of the local state $\rho^* =
\sum_{i=1}^D p_i \dyad{E_i}$. From \eqref{eq:trabajoCorrGeneralCota2}
we have that
\begin{align}
  \Wmin (\rho^*, N) & \leq N\langle E \rangle + N \kBT \log \ZS -\kBT \log \left[
	g_N(N\langle E \rangle) \right] \nonumber \\
  & \leq N\langle E \rangle + N\kBT\log \ZS - \kBT S(\rho^*) + \mathcal{O}(\log
  N) \nonumber \\
  &= N \left[ \langle E \rangle - \kBT S(\rho^*) \right] + N\kBT \log \ZS +
	\mathcal{O}(\log N) \nonumber \\
  &= N \left[ F(\rho^*) - F(\tau) \right] + \mathcal{O}(\log N)  .
  \label{eq:upperbound}
\end{align}
Finally, like in the case $D=2$, using \eqref{eq:capitulo3cotaW1Winf} we
can obtain the lower bound
\begin{equation}
	N \DF(\rho^*) \leq \Wmin (\rho^* , N) \leq N \DF(\rho^*) + \mathcal{O}(\log N).
\end{equation}

If there exists $d$ such that $p_dN \not \in \mathbb{N}$, we just replace $\langle E \rangle$ with an approximation of the mean energy $\langle \tilde E \rangle$ with $| \langle \tilde E \rangle - \langle E \rangle| < D/N$ and the probabilities $p_d$ with $\tilde p_d$ with $| \tilde p_d - p_d | < 1 / N$. These constructions were provided in the previous section  when we analyzed the density of  reversible states. Thus, when $N \to \infty$ we have $\langle \tilde E \rangle \to \langle E \rangle$ and $\tilde p_d \to p_d$ and we recover the same results.  

The upper bound of \eqref{eq:upperbound}, while it was derived for optimal reversible states, 
is rather general. It is valid in the large $N$ limit for every multipartite state  that: 
$(i)$ is reversible, $(ii)$ satisfies the partial trace condition, and $(iii)$ it has the mean energy level in its support. 
In particular, we have derived a family of states $\sigma_d$ in \eqref{eq:sigmad} that fulfill these three properties for 
every local state $\rho$. Since these states are not optimal, we have that $\Wmin (\rho^*, N) \le W_{\rm form}(\sigma_d)$
and also  $W_{\rm form}(\sigma_d) \leq N \Delta F(\rho) + \mathcal{O}(\log N)$. Thus, for every state $\rho$ we have that:
\begin{equation}
	\DF(\rho) \leq \frac{\Wmin(\rho, N)}{N} \leq \DF(\rho) +
	\mathcal{O}\left( \frac{\log N}{N} \right).
\end{equation}
Therefore, the \cwork~of any state converges to the standard thermodynamic
free energy difference in the thermodynamic limit, $N \to \infty$, at a rate
$\mathcal{O}\left(\log(N)/N\right)$.

On the other hand, while we cannot obtain the analytical expression of the optimal state in order to obtain the extractable work
(since it depends on the specific spectra of each case) 
we have shown that there exist a family of reversible states that in the thermodynamic limit attain the minimum work
cost per particle, recovering standard results from thermodynamics.


\section{Generalization to different local states}

Up to now we have limited ourselves to study the task of creating $N$ copies
of the same local state $\rho$. An obvious generalization of the previous
results is to allow each of the $N$ subsystems to have a different local state with an arbitrary
Hamiltonian.
That is to say, given a set of local states
$\left\{\rho_i\right\}_{i=1,\dots,N}$ (of arbitrary different dimensions
$\left\{d_i\right\}_{i=1,\dots,N}$ and with arbitrary local Hamiltonian
$\left\{H_i\right\}_{i=1,\dots,N}$), we are now interested in creating a global
state $\rho^{(N)}$ such that
\begin{equation} \label{eq:constraints_assym}
    \tr_{-i}\left(\rho^{(N)}\right) = \rho_i, \quad i = 1,\dots,N .
\end{equation}
Furthermore, just like before, it is of particular interest to find the state
$\rho^{(N)}_{\rm{min}}$ that minimizes the work of formation among all the
states that satisfy the constraints of \eqref{eq:constraints_assym}. Now the
optimal state will in general not belong to the subset $\mathcal{C}^*$, given
that our system is no longer  symmetric under permutations.
Nonetheless, the optimal solution still is very similar to the
previous one.

\begin{CorollarySM}[Corollary of Theorem \ref{teo:generalizacion}]
    The state $\rho^{(N)}_{\rm{min}}$ that satisfies the constraints of
    \eqref{eq:constraints_assym} and minimizes the work of formation has the
    form
    \begin{equation} \label{eq:solucion_teorema_general_assym}
        \left[\rho^{(N)}_{\rm{min}}\right]_i =
        \gamma \left[\tau\right]_i \mathbbm{1}_{ \{ i \in U \} } +
        \sum_{j=1}^{M} \gamma s_{j} \left[\tau\right]_i
        \mathbbm{1}_{ \{i = n_j \} },
        \quad i \in \left\{1, 2, \dots, D_N\right\},
 	  \end{equation}
    where
    $\tau = \tau_1 \otimes \tau_2 \otimes \dots \otimes \tau_N$ with
    $\tau_k = e^{-\beta H_k}/Z_k$, $k = 1,\dots,N$ is the $N$-partite thermal state;  
    $D_N = \prod_{i=1}^{N} d_i$;
    $M \leq \sum_{i=1}^N d_i - N + 1$;
    $s_{1}, s_{2}, \ldots, s_{M} \in [0,1]$;
    $\gamma \in \R_{> 0}$;
    $U$ is a set of indexes where
    $\left\{\left[\rho^{(N)}_{\rm{min}}\right]_i/\left[\tau\right]_i\right\}_{i=1,\dots,D_N}$
    realizes its infinity norm; and
    $n_1, n_2, \dots, n_{M} \in \{1, 2, \dots, D_N\}$.
\end{CorollarySM}
\begin{proof}
    This result is immediate from Theorem \ref{teo:generalizacion}. Indeed we
    now seek to minimize the infinity norm of a vector $q$ given by
    \begin{equation}
        q_i = \frac{\left[\rho^{(N)}_{\rm{min}}\right]_i}{\left[\tau\right]_i}
    \end{equation}
    subject to the linear constraints of the partial trace given in \eqref{eq:constraints_assym}.

    Notice now that in the proof of Theorem \ref{teo:generalizacion}, the
    details of the matrix $A$, that codifies the linear constrains on the vector
    whose infinity norm we are minimizing, is irrelevant in the proof of the
    general solution \eqref{eq:solucion_teorema_general}. Therefore,
    the same solution must also hold for the linear constraints
    \eqref{eq:constraints_assym} in the case where each local system is different.
\end{proof}

\bigbreak

Similarly to the previous cases, again we will have particular cases where the
optimal state is reversible. This will happen whenever the local states are such
that in the optimal solution we have that $s_1 = s_2 = \dots = s_M = 0$
or $s_1 = s_2 = \dots = s_M = 1$.


	\subsection{Thermodynamic limit}

In this last section we are going to prove a generalization of our results in the thermodynamic limit. 

	\begin{TheoremSM}[Theorem 3 in the main text]
	Let $(p^{(1)}, E^{(1)}), (p^{(2)}, E^{(2)}), \ldots , (p^{(N)}, E^{(N)}) \in \mathbb R^D \otimes \mathbb R^D$  be an i.i.d sample with arbitrary distribution  $\mathcal D$ such that the maximum level of energy is upper bounded, and $\mathcal W_N$ the \cwork~of a system with block-diagonal reduced states $\rho_i$ defined by the probability vector $p^{(i)}$ and Hamiltonian
	with energies $E^{(i)}$. Then,
	\begin{equation}
    \frac{\mathcal W_N}{N} \xrightarrow{N \rightarrow \infty} \mean{ \Delta F }_{\mathcal D},
	\end{equation}
	where the mean in $\Delta F$ is with respect to $\mathcal D$ and the convergence is almost surely.
	\label{teo:asintotico_ae}
	\end{TheoremSM}
The almost surely convergence in Theorem \ref{teo:asintotico_ae} means that given $\epsilon > 0$, there exists $N_0 < \infty$	such that 
	\begin{equation}
	\left | \frac{\mathcal W_N}{N} -\mean{ \Delta F }_{\mathcal D} \right | < \epsilon \qquad \text{ for all $N > N_0$.}
	\end{equation}
	Notice that $N_0$ is also a random variable that takes a finite value. 

	\begin{proof}
	First, we are going to analyze the case where the support of $\mathcal D$ is discrete and finite. The extension to absolutely continuous, singular continuous and discrete distribution with numerable support follow from the first case. Notice that these cases allows us to say that   $\mathcal D$ is an arbitrary distribution (see Lebesgue's decomposition theorem) and then we are not making any extra assumption on $\mathcal D$ except that $\mean{ \Delta F }_{\mathcal D}$ is well defined. 
	
	Let's consider first the case where $\mathcal D$ consists of a discrete
        distribution with support in a finite set
    \begin{equation}
        \{ (p_1^\star , E_1^\star), (p_2^\star , E_2^\star), \ldots, (p_K^\star, E_K^\star) \} \subset \mathbb R^{D} \times \mathbb{R}^D,
    \end{equation}
    and denote the probability of taking each one of these values as 
	\begin{equation}
	\mathbb P \left( p_i = p_j^\star, E_i = E_j^\star \right) = r_j, \quad \sum_{j=1}^K r_j = 1, \quad r_j > 0.  
	\end{equation}
	If $n_j := \sum_{i=1}^N \mathbb I_{ \{ p_i = p_j^\star, E_i = E_j^\star \} }$ is the number of subsystems with eigenvalues $p_j^\star$ and respective energies $E_j^\star$, then $(n_1, \ldots, n_K)$ is distributed as a multinomial distribution with $N$ numbers of trials and parameter $(r_1, \ldots, r_K)$. Given $\epsilon_2 > 0$, let us consider the event
		\begin{equation}
		\Omega_{\epsilon_2, N} = \bigg \{  (1-\epsilon_2)r_j N \leq n_j \leq (1+\epsilon_2)r_jN \text{ for all } j=1,\ldots, K  \bigg \}.
		\end{equation}
	Using the Hoeffding inequality we can prove that the probability of the complement of $\Omega_{\epsilon_2, N}$ is bounded as
		\begin{align}
		\mathbb P \left(	 \Omega_{\epsilon_2, N}^c \right) 
		& =
		\mathbb P \left(  \bigcup_{j=1}^K \bigg \{ \left| n_j - Nr_j  \right| > {\epsilon_2 r_j N} \bigg \} \right) \nonumber \\
		&\leq
		\sum_{j=1}^K \mathbb P \bigg (   \left| n_j - Nr_j  \right| > {\epsilon_2 r_j N} \bigg ) \nonumber \\
		& \leq
		\sum_{j=1}^K 2 e^{ -2 \epsilon^2_2  r_j^2 N  },
		\end{align}
	that is, except for a probability exponentially small in $N$ the system satisfies $\Omega_{\epsilon_2, N}$. 
	On the other hand, on $\Omega_{\epsilon_2, N}$  for $\rho_j^\star = \sum_{d=1}^D (p_j^\star)_d \dyad{(E^\star_j)_d}$ it holds that 
	\begin{align}
	\mathcal W_N 
	& \leq 
	\sum_{j=1}^K \Wmin ( n_j, \rho_j^\star ) \nonumber \\
	& \leq
	\sum_{j=1}^K (1+\epsilon_2) r_j N \Delta F (\rho_j^\star) + \mathcal{O}(\log n_j) \nonumber \\
	&= (1 + \epsilon_2) N \mean{ \Delta F }_{\mathcal D} + \mathcal O (K \log N),
	\end{align}
	where we use the result from the thermodynamics limit for copies. Notice that the correlations that are needed in order to bound $\mathcal W_N$ just include the identical subsystems. 
		In the same way from the subadditivity of the von Neumann entropy it follows that  $(1-\epsilon_2)N \mean{ \Delta F }_{\mathcal D} < \mathcal W_N$ on $\Omega_{\epsilon_2, N}$. Let's consider $\epsilon_2$ such that $2 \epsilon_2 \mean{\Delta F}_{\mathcal D} < \epsilon $. On the other hand, since the $\mathcal O (K \log N)$ is just an additive constant, there exists $N_0 = N_0 (\epsilon)$ such that it is upper bounded by $\epsilon N / 2$ for all $N > N_0$. Finally, 
	\begin{equation}
	\left | \frac{\mathcal W_N}{N} - \mean{ \Delta F }_{\mathcal D} \right | < \epsilon \qquad \text{ for all $N > N_0$ on $\Omega_{\epsilon_2, N}$}.
	\end{equation}	 
	Since $\sum_{N=1}^\infty \mathbb P ( \Omega_{\epsilon_2, N}^c ) < \infty$, from the Borel-Cantelli lemma it follows that $\mathcal W_N / N$ converges almost surely to the mean value of the free energy difference. 
	
	Now we are going to generalize the proof for a more generic class of distributions $\mathcal D$. Let $\{ p^{(i)}, E^{(i)} \}_{i=1, 2,\ldots, N}$ be distributed according to a general distribution $\mathcal D$ on $\mathbb R^{2D}_{\geq0}$. The distribution in energy is upper bounded by a maximum energy $E_{\rm max}$. From the subadditivity of the von Neumann entropy is clear that
	\begin{equation}
	\mean{ \Delta F }_{\mathcal D} \leq \lim_{N \rightarrow \infty} \frac{\mathcal W_N}{N}. 
	\end{equation}
	Let's prove the other inequality. 
	First, notice that if we consider a disjoint partition $V_1, V_2, \ldots V_R$ of the support of the distribution $\mathcal D$, then it is enough to prove that 
	\begin{equation}
	\frac{ \mathcal W_N^{V_j}}{N} \xrightarrow{N \rightarrow \infty} \mathbb E [ \Delta F | V_j ] \mathbb P (V_j) \quad \text{almost surely for all $j =1, 2, \ldots, R$,}
	\end{equation}
	where $\mathcal W^{V}_N$ represents the correlated work of formation of those subsystems with probabilities and energies on $V$; $\mathbb E [ \Delta F | V ]$ is the expectation value of the difference of free energy conditional to $V$; and $\mathbb P (V)$ is the probability under $\mathcal D$ of $V$. This follows from the fact that 
	\begin{equation}
	\frac{\mathcal W_N}{N}  - \mean{\Delta F}_{\mathcal D}
	\leq 
	\sum_{j=1}^R \frac{\mathcal W_N^{V_j}}{N} - \mathbb E [ \Delta F | V_j ] \mathbb P (V_j)	.
	\end{equation}
Let us consider the function $\psi$ from $ \{ 0, 1 \}^D$ to the sets of $\mathbb R^D$ defined as $\psi (a_1, \ldots ,a_D) = \{ p \in \mathbb R^{D} : p_i = 0 \text{ if } a_i = 0  \text{ and } p_i > 0 \text{ if } a_i = 1 \} $. Notice that the family of sets $\mathcal V$ with elements  $\psi(a_1, \ldots, a_D) \times [0, E_{\rm max}] $ for $a_i = 0, 1$ and $a_1 + \ldots + a_D \geq 1$ defines a disjoint partition of the support of $\mathcal D$. In particular, if $\mathcal D$ is an absolute continuous distribution, the only set of this family with non-zero probability is $V_0 := \{ (p,E) \in \mathbb R^{2D} : \min_i p_i > 0  \}$. For the sake of explanation, let us consider this case and then we generalize our results to the case where more elements with non-zero probability are present on $\mathcal V$. 
	
	Since $\mathbb P (V_0) = 1$, we can omit the conditionals $V_0$ in each one of the previous expressions. Given $n \in \mathbb N$, we define 
	\begin{equation}
	K_n = \left \{  (p,E) \in \mathbb R^{2D} : \min_{i} p_i \geq 1/n  \right \}.
	\end{equation}
	Then we have that $\mathbb P (K_n) \rightarrow 1$ as $n \rightarrow \infty$. Given $\epsilon_3 > 0$, let $n$ be such that $\mathbb P (K_n^c) < \epsilon_3$. Now, we can define $\Delta F_{\rm max} = \max_{\rho} \Delta F (\rho)$ where the maximum is taken respect all the possible states $\rho$ with energy distribution bounded by $E_{\rm max}$. Since 
	\begin{equation}
	\mean{\Delta F}_{\mathcal D} = \mathbb E \left [ \Delta F | K_n \right ] \mathbb P (K_n) + \mathbb E \left [ \Delta F | K_n^c \right ] \mathbb P (K_n^c),
	\end{equation}
	it follows that
	\begin{equation}
	\mean{\Delta F}_{\mathcal D} -  \mathbb E \left [ \Delta F | K_n \right ] \mathbb P (K_n)
	\leq
	\epsilon_3 \Delta F_{\rm max}.
	\end{equation}
	On the other hand, $\mathcal W_N \leq \mathcal W_N^{\mathbb K_n} + \mathcal W_N^{\mathbb K_n^c}$, where $\mathcal W_N^{\mathbb K_n}$ is the correlated work of formation of the systems with probabilities and energies in $\mathbb K_n = \{ (p^{(i)}, E^{(i)}) : (p^{(i)}, E^{(i)}) \in K_n \}$. Notice that in the set $ \Omega_{\epsilon_3} = \{ \# \mathbb K_n^c < 2 \epsilon_3 N \}$ it holds  $\mathcal W_N^{\mathbb K_n^c} \leq 2 \epsilon_3 N W_{\text{form}}^{\rm max}$, where $W_{\text{form}}^{\rm max}$ is the maximum work of formation of all the states $\rho$ with energy distribution bounded by $E_{\rm max}$. Given $\epsilon > 0$, if we choose $\epsilon_3 = \epsilon / 2 ( 2 W_\text{form}^{\rm max} + \Delta F_{\rm max} )$, then
	\begin{equation}
	\frac{\mathcal W_N}{N} - \mean{\Delta F}_{\mathcal D} < \frac{\mathcal W_N^{\mathbb K_n}}{N} - \mathbb E \left [ \Delta F | K_n \right ] \mathbb P (K_n) + \frac{\epsilon}{2} \quad \text{on $\Omega_{\epsilon_3}$.}
	\end{equation}
	On the other hand, 
	\begin{align}
	\frac{\mathcal W_N^{\mathbb K_n}}{N} - \mathbb E \left [ \Delta F | K_n \right ] \mathbb P (K_n) 
	&=
	 \frac{\mathcal W_N^{\mathbb K_n}}{\# \mathbb K_n}\frac{\# \mathbb K_n}{N} - \mathbb E \left [ \Delta F | K_n \right ] \mathbb P (K_n)  \nonumber \\
	&< 
	\left ( \frac{\mathcal W_N^{\mathbb K_n}}{\# \mathbb K_n} - \mathbb E \left [ \Delta F | K_n \right ]  \right )
	+ 
	\left ( \mathbb E \left [ \Delta F | K_n \right ] \left( \frac{\# \mathbb K_n}{N} - \mathbb P (K_n) \right) \right )	.
	\label{eq:desacoplo}
	\end{align}		
	Since $\# \mathbb K_n / N \rightarrow \mathbb P (K_n)$ almost surely, there exists $N_2 < \infty$ such that the second term in \eqref{eq:desacoplo} is smaller than $\epsilon / 4$ for all $N > N_2$. On the contrary, the first term corresponds to the thermodynamic limit for the subset of systems in $\mathbb K_n$. In order to use the result for discrete and finite distributions we need the following lemma. The proof of the lemma is given at the end of the supplementary material.
		 \begin{LemmaSM}[Continuity of Thermo-majorization]
	 Given $w > 0$, and a state $\rho(p,E) := \sum_{i=1}^D p_i \dyad{E_i}$, with $(p,E) \in \mathbb R^{2D}$, there exists $\delta = \delta(w, p, E)$ such that for any other $(\tilde p, \tilde E) \in \mathbb R^{2D}$ 
	 \begin{itemize}
	 \item if $\|  (p,E) - (\tilde p, \tilde E) \| < \delta$  then the transformation 
	 \begin{equation}
	 \rho(p, E) \otimes \dyad{w}   \Ttransform \rho(\tilde p, \tilde E) \otimes \dyad{0}  
	 \end{equation}
	 can be performed with thermal operations;
	 \item if $\|  (p,E) - (\tilde p, \tilde E) \| < \delta$ and $p_i = 0$ implies $\tilde p_i = 0$, then the transformation 
	 \begin{equation}
	 \rho(\tilde p, \tilde E) \otimes \dyad{w} \Ttransform \rho(p,E) \otimes \dyad{0}
	 \end{equation}
	 can be performed with thermal operations. 
	 \end{itemize} 
	 \label{lemma:continuity-majorization}
	  \end{LemmaSM} 
In the previous lemma, the distance measure between probabilities and energies could be chosen as any of the vector norms in $\mathbb R^{2D}$, for example, the infinity norm given by $\| (p,E) - (\tilde p , \tilde E) \|_\infty := \max \{ \| p - \tilde p \|_\infty , \| E - \tilde E \|_\infty \}$. On the other hand, notice that the condition $\tilde p_i = 0$ if $p_i = 0$ is equivalent to the condition of $\text{supp}( \rho(\tilde p, \tilde E))  \subseteq \text{supp} ( \rho (p, E) )$ when $E = \tilde E$.

	With the previous considerations, we need to prove the theorem for a distribution supported on $K_n$ with $n = n(\epsilon)$ chosen as before. Let's consider $w = \epsilon / 16$ and for each pair $(p, E) \in K_n$ we select $\delta_2(p, E)$ as in Lemma \ref{lemma:continuity-majorization}. On the other hand, $\Delta F (\rho)$ is a continuous function of $(p,E)$. Since $K_n$ is a compact set, $\Delta F(\rho)$ is also absolutely continuous and there exists $\delta_0$ such that if $\| (p,E) - (\tilde p, \tilde E) \| < \delta_0$ then $| \Delta F(\rho) - \Delta F (\tilde \rho) | < \epsilon / 16$. Let's define $\delta (p, E) = \min \{ \delta_0, \delta_2 (p,E) \}$. Then we can cover $K_n$ with the family of open balls with center in $(p, E)$ and radius $\delta (p, E)$ as
	\begin{equation}
	K_n \subset \bigcup_{(p,E) \in K_n} B ( (p,E) , \delta (p,E) ).
	\end{equation}	 
	Since $K_n$ is a compact set on $\mathbb R^{2D}$, we can select a finite family $\{ (p^*_j , E^*_j) \}_{j=1, 2, \ldots, K}$ such that
	\begin{equation}
	K_n \subset \bigcup_{j=1}^K B \left( (p^*_j , E^*_j), \delta (p^*_j, E^*_j) \right).
	\end{equation}
	Let us define the family of sets
  \begin{align}
    V_1 &= B \left( (p^*_1 , E^*_1), \delta (p^*_j, E^*_j) \right) \cap K_n, \\
    V_j &= B \left( (p^*_j , E^*_j) , \delta (p^*_j, E^*_j) \right) \cap K_n
    \setminus \cup_{l=1}^{j-1} B \left( (p^*_l , E^*_l) , \delta (p^*_l, E^*_l)
    \right), \quad j = 2,\dots,K,
  \end{align}
  and their respective probabilities $r_j = \mathbb P (V_j)$. Notice that $V_1, V_2, \ldots V_K$ is a partition of $K_n$. If $\tilde{ \mathcal{D}}$ denotes the discrete and finite distribution with support in $\{ (p^*_j , E^*_j) \}_{j=1, 2, \ldots, K}$ and pointwise probabilities given by $r_1, r_2, \ldots, r_K$, then  
	\begin{align}
	\frac{\mathcal W_N^{\mathbb K_n}(\mathcal D)}{\# \mathbb K_n} - \mathbb E_{\mathcal D} \left [ \Delta F | K_n \right ] 
	& < 
	\left ( \frac{\mathcal W_N^{\mathbb K_n}(\tilde{\mathcal{D}})}{\# \mathbb K_n} - \mathbb E_{\tilde{\mathcal{D}}} \left [ \Delta F | K_n \right ]  \right )
	+
	\left ( \frac{\mathcal W_N^{\mathbb K_n}(\mathcal{D})}{\# \mathbb K_n} - \frac{\mathcal W_N^{\mathbb K_n}(\tilde{\mathcal{D}})}{\# \mathbb K_n}   \right )
	+	
	\left ( \mathbb E_{\tilde{\mathcal{D}}} \left [ \Delta F | K_n \right ] - \mathbb E_{D} \left [ \Delta F | K_n \right ]  \right )	\nonumber 	\\
	& <
	\left | \frac{\mathcal W_N^{\mathbb K_n}(\tilde{\mathcal{D}})}{\# \mathbb K_n} - \mathbb E_{\tilde{\mathcal{D}}} \left [ \Delta F | K_n \right ]  \right | 
	+ w + \frac{\epsilon}{16}.
	\end{align}
	In the last equation we have used that ${\mathcal W_N^{\mathbb K_n}(\mathcal{D})} -{\mathcal W_N^{\mathbb K_n}(\tilde{\mathcal{D}})}   < \# \mathbb K_n w$. In order to arrive at this bound we use Lemma~\ref{lemma:continuity-majorization}, which tells us that we can transform a state defined by $(p^*_j,E_j^*)$ into any other state in $B((p^*_j,E^*_j),\delta(p^*_j,E_j^*))$ using an amount of work $w$.  Now, we can consider the following two $\# \mathbb K_n$-partite optimal states,
	the first one $\rho_{\tilde{\mathcal D}}^{(\# \mathbb K_n)}$ is defined by the discrete distribution  $\tilde{ \mathcal D}$ with support in $\{ (p^*_j , E^*_j) \}_{j=1, 2, \ldots, K}$ and  work of formation $\mathcal W_N^{\mathbb K_n}(\tilde{\mathcal{D}})$; the other $\rho_{{\mathcal D}}^{(\# \mathbb K_n)}$ is defined by the distribution $\mathcal D$ with work of formation $\mathcal W_N^{\mathbb K_n}({\mathcal{D}})$. Let us also assume that 
	$\mathcal W_N^{\mathbb K_n}({\mathcal{D}})\ge\mathcal W_N^{\mathbb K_n}(\tilde{\mathcal{D}})$ (if this is not the case we can bound the above term with $w=0$). Then, notice that
	with an amount of work $(\# \mathbb K_n\, w)$ we can transform  
	$\rho_{\tilde{\mathcal D}}^{(\# \mathbb K_n)} \rightarrow \tilde \rho_{{\mathcal D}}^{(\# \mathbb K_n)}$, this is done
	by applying locally to each subsystem the thermal operation that transform with work $w$  the state with $(p^*_j,E_j^*)$ to the corresponding state $(p,E)$ defined by $\mathcal D$. This transformation guarantees that the state  $\tilde \rho_{{\mathcal D}}^{(\# \mathbb K_n)}$ is locally equivalent to the optimal state $\rho_{{\mathcal D}}^{(\# \mathbb K_n)}$, and allows us to bound $ \mathcal W_N^{\mathbb K_n}({\mathcal{D}})$.  Thus,  $\mathcal W_N^{\mathbb K_n}(\tilde{\mathcal{D}}) + (\# \mathbb K_n\, w)\ge \tilde{\mathcal W}_N^{\mathbb K_n}(\tilde{\mathcal{D}}) \ge \mathcal W_N^{\mathbb K_n}({\mathcal{D}}) $, with $\tilde{\mathcal W}_N^{\mathbb K_n}({\mathcal{D}})$ the work
	of formation of $\tilde \rho_{{\mathcal D}}^{(\# \mathbb K_n)}$.

	Finally, all that remains is to bound the first term which corresponds to the thermodynamic limit for the discrete distribution. By the previous results, this can be done by taking $N_1$ such that 
	\begin{equation}
	\left | \frac{\mathcal W_N^{\mathbb K_n}(\tilde{\mathcal{D}})}{\# \mathbb K_n} - \mathbb E_{\tilde{\mathcal{D}}} \left [ \Delta F | K_n \right ]  \right | 
	< \frac{\epsilon}{8} \quad \text{for all $N > N_1$}.
	\end{equation}
	
	Notice that the main ingredient of the proof was that we can restrict the analysis to a compact subset with elements that have the same support (in order to use Lemma \ref{lemma:continuity-majorization}). The same analysis can been performed in each set in $V \in \mathcal V$ with the topology of the open sets restricted to $V$.
	\end{proof}
	 
  \bigbreak
	  
	 \begin{proof}[Proof of Lemma \ref{lemma:continuity-majorization}]
	 Let us consider  $E_1\leq E_2,...,\leq E_D$ and  a $\beta$-order $\pi (1), \pi(2), \ldots, \pi(D)$ of $1, 2, \ldots, D$ such that  
	\begin{equation}	 
	 p_{\pi(1)} e^{\beta E_{\pi(1)}} \geq p_{\pi(2)} e^{\beta E_{\pi(2)}} \geq \ldots \geq p_{\pi(D)}e^{\beta E_{\pi(D)}}.
	 \end{equation}
Given another $(\tilde p, \tilde E)$, we can choose $\delta_1 > 0$ such that if $\| (p, E) - (\tilde p, \tilde E)  \| < \delta_1$ then the same ordering holds, that is,  
	\begin{equation}
\tilde p_{\pi(1)} e^{\beta \tilde E_{\pi(1)}} \geq \tilde p_{\pi(2)} e^{\beta \tilde  E_{\pi(2)}} \geq \ldots \geq \tilde p_{\pi(D)}e^{\beta \tilde E_{\pi(D)}}. 
	\end{equation}
If all the $\geq$ are actually $>$, then this follows directly selecting $\delta_1$ small enough. On the other hand, if there exists $i$ such that $p_{\pi (i)} e^{\beta E_{\pi(i)}} = p_{\pi (i+1)} e^{\beta E_{\pi(i+1)}}$, we can change the ordering $\pi$ for another $\tilde \pi$ with $\tilde \pi (i)  = \pi(i+1)$ and $\tilde \pi(i+1) = \pi(i)$ in order to ensure the ordering. For simplicity, in the following analysis we are going to use $i$ instead of $\pi (i)$.
	 
	 The necessary and sufficient condition to ensure that the transformation $\rho(p, E) \otimes \dyad{w}  \Ttransform \rho(\tilde p, \tilde E) \otimes \dyad{0} $ is possible via thermal operations is \deit{thermo-majorization} \cite{Horodecki2013}, which means that the polygonal defined by the vertices on	$(x_0, y_0) = (0,0), (x_1, y_1), \ldots, (x_D, y_D)$, with $x_i =e^{-\beta w}(e^{-\beta E_1} + \ldots + e^{-\beta E_i })$ and $y_i = p_1 + \ldots + p_i$, lies completely above the polygonal with vertices on $(\tilde x_0, \tilde y_0) = (0,0), (\tilde x_1, \tilde y_1), \ldots, (\tilde x_D, \tilde y_D)$, with $\tilde x_i = e^{-\beta \tilde E_1} + \ldots + e^{-\beta \tilde E_i }$ and $y_i =  \tilde p_1 + \ldots + \tilde p_i$.
	 If we choose $\| E - \tilde E \|_\infty  < w$, then it holds
	 \begin{equation}
	 e^{-\beta w} \left( e^{-\beta E_1} + \ldots + e^{-\beta E_i} \right)  < e^{-\beta \tilde E_1} + \ldots + e^{-\beta \tilde E_i},
	 \end{equation}
	 for all $i = 1, 2, \ldots , D$,
	 which is equivalent to $x_i < \tilde x_i$. Equivalently, if $w$ is such that 
	 \begin{equation}
	 w < \delta_2 :=  \frac{1}{2}  \kBT \max_{i=1, \ldots, D} \log \left(  1 + \frac{e^{-\beta E_{i}}}{ e^{-\beta E_1} + \ldots + e^{-\beta E_D} }  \right),
	 \label{eq:w_acotado}
	 \end{equation}
	then $\tilde x_i < x_{i+1}$. Notice that the condition in \eqref{eq:w_acotado} is not restrictive.   

	With the previous considerations, if we choose $\delta < \min \{ \delta_1, \delta_2, w\}$ then the thermo-majorization condition will be satisfied if each $(\tilde x_i , \tilde y_i)$ lies below the line that connects $(x_i, y_i)$ with $(x_{i+1}, y_{i+1})$, that is
	\begin{equation}
	\tilde y_i \leq y_i + \frac{y_{i+1}-y_i}{x_{i+1}-x_i}(\tilde x_i - x_i).
	\label{eq:linear_majorization}
	\end{equation}
	For those $i$ such this $p_{i+1} \neq 0$, that is equivalent to
	\begin{equation}
	\frac{e^{-\beta E_{i+1}}}{p_{i+1}}\sum_{j=1}^i (\tilde p_j - p_j) \leq 
	e^{\beta w} \left( e^{-\beta \tilde E_1} + \ldots + e^{-\beta \tilde E_i}  \right) - \left( e^{-\beta E_1} + \ldots + e^{-\beta E_i} \right) . \label{eq:condition_thermomajorization}
	\end{equation}
	The condition of \eqref{eq:condition_thermomajorization} will be simultaneously satisfied for all $i$ with $p_i \neq 0$ if
	\begin{equation}
	 e^{\beta \| E - \tilde E \|_\infty} \left ( \frac{ e^{\beta E_1}}{ \min_{i : p_i \neq 0} p_i e^{\beta E_i} } \| p - \tilde p \|_1 + 1\right) \leq e^{\beta w}.
	\label{eq:condition_delta3}
	\end{equation}

	Then, we can choose $\delta = \min \{  \delta_1, \delta_2, \delta _3, w \}$ with $\delta_3$ such that if $\| p - \tilde p \|_1 < \sqrt{D} \| p - \tilde p \|_\infty < \delta_3 $ and $\| E - \tilde E \|_\infty < \delta_3$ then \eqref{eq:condition_delta3} holds. When the probability vectors are equal, it is just the maximum energy difference that bounds the work $w$. On the other hand, if $p_{i+1} = 0$ then \eqref{eq:linear_majorization} is trivially satisfied since $y_i = 1$. 
	
	Let's now study the transformation $\rho(\tilde p, \tilde E) \otimes \dyad{w} \Ttransform \rho(p,E) \otimes \dyad{0}$. Just like before, we can choose $\delta_1 > 0$ such that if $\| (p,E) - (\tilde p, \tilde E) \|_1 < \delta_1$, the thermo-majorization order does not change. The first polygonal now is defined as $(x_0 , y_0) = (0,0), (x_1, y_1) , \ldots, (x_D, y_D)$ with $x_i = e^{-\beta w} ( e^{-\beta \tilde E_1} + \ldots + e^{-\beta \tilde E_i} )$ and $y_i = \tilde p_1 + \ldots + \tilde p_i$; and the other (which has to lie completely below the first one) by $(\tilde x_0, \tilde y_0) = (0,0), (\tilde x_1, \tilde y_1), \ldots, (\tilde x_D, \tilde y_D) $ with $\tilde x_i = e^{-\beta E_1} + \ldots + e^{-\beta E_i}$ and $\tilde y_i = p_1 + \ldots + p_i$. Proceeding like before, there exists $\delta_2 > 0$ such that if $\| E - \tilde E \|_\infty < \delta_2$ then $\tilde x_{i-1} < x_i < \tilde x_{i}$ for all $i = 1, 2, \ldots, D$.
	
	If $\tilde x_{i-1} < x_i < \tilde x_{i}$, then a sufficient condition is that $(x_i, y_i)$ lies above the line that connects $(\tilde x_i, \tilde y_i)$ with $(\tilde x_{i+1}, \tilde y_{i+1})$, that is,
	\begin{equation}
	y_i \geq \tilde y_i - \frac{\tilde y_{i+1} - \tilde y_i}{ \tilde x_{i+1} - \tilde x_i } (\tilde x_i - x_i).
	\end{equation}
	If $p_{i+1} \neq 0$, this is equivalent to 
	\begin{equation}
	\frac{e^{-\beta E_{i+1}}}{p_{i+1}} \sum_{j=1}^i (p_j - \tilde p_j)
	\leq 
	\left(  e^{-\beta E_1} + \ldots + e^{-\beta E_i}  \right) - e^{\beta w} \left( e^{-\beta \tilde E_1 }+ \ldots + e^{-\beta \tilde E_i} \right).
	\label{eq:majoriza111}
	\end{equation}
	A sufficient condition to ensure that \eqref{eq:majoriza111} holds is \eqref{eq:condition_delta3}. On the other hand, if $p_{i+1} = 0$ then we need $y_i = \tilde p_1 + \ldots + \tilde p_i \geq \tilde y_i = 1$, which is satisfied if and only if $\tilde p_j = 0$ for all $j > i$. This is true by hypothesis that the support of $\tilde p$ is contained in the support of $p$ and the fact that $p_{i+1} = 0$ implies that $p_{j} = 0$ for all $j \geq i + 1$, since $p_i$ is $\beta$-ordered.
	 \end{proof}






\end{document}